\documentclass[journal,final,romanappendices]{IEEEtran}

\usepackage{amsmath}
\usepackage{amssymb}
\usepackage{amsfonts}
\usepackage{amscd}
\usepackage{amsthm}
\usepackage[noadjust]{cite}
\usepackage{mathrsfs}
\usepackage{amsbsy}

\usepackage{tcolorbox}

\usepackage{tabulary}
\usepackage{makecell}

\newtheorem{theorem}{Theorem}
\newtheorem{lemma}{Lemma}

\newtheorem{corollary}{Corollary}

\theoremstyle{definition}
\newtheorem{example}{Example}
\newtheorem{definition}{Definition}
\newtheorem{remark}{Remark}

\makeatletter
\renewcommand*\env@matrix[1][c]{\hskip -\arraycolsep
  \let\@ifnextchar\new@ifnextchar
  \array{*\c@MaxMatrixCols #1}}
\makeatother
\usepackage{graphicx}

\usepackage{etoolbox}
\AtEndEnvironment{remark}{\null\hfill\IEEEQED}%
\AtEndEnvironment{example}{\null\hfill\IEEEQED}%
\AtEndEnvironment{definition}{\null\hfill\IEEEQED}%

\newcommand{\Fb}{\mathbb{F}} 
\newcommand{\tth}{{\text{th}}} 
\newcommand{\rx}{{\sf Rx}}
\newcommand{\tr}{{\sf T}} 
\newcommand{\ravg}{{r_{\rm avg}}}
\newcommand{\Db}{{\mathcal{D}}}
\newcommand{\Kb}{{\mathcal{K}}}
\newcommand{\Sc}{{\mathcal{S}}}
\newcommand{\Mc}{{\mathcal{M}}}
\newcommand{\supp}{{\sf supp}}
\newcommand{\spp}{{\sf span}}
\newcommand{\colsp}{{\mathcal{C}}}
\newcommand{\nullsp}{{\mathcal{N}}}
\newcommand{\rank}{{\rm rank}}
\newcommand{\minrank}{{{\rm minrk}_q}}
\newcommand{\Vin}{{V_{\sf in}}}
\newcommand{\Vout}{{V_{\sf out}}}

\newcommand{\p}{\pmb}
\newcommand{\xb}{\pmb{x}} 
\newcommand{\cb}{\pmb{c}}

\newcommand{\Ac}{\mathcal{A}} 
\newcommand{\Bc}{\mathcal{B}} 
\newcommand{\Vc}{\mathcal{V}} 
\newcommand{\Ec}{\mathcal{E}}

\newcommand{\Zc}{\mathcal{Z}} 
\newcommand{\Cs}{\mathscr{C}} 
\newcommand{\Ds}{\mathscr{D}} 

\newcommand{\Gcpu}{\bar{G}_u} 
\newcommand{\Ecpu}{\bar{\mathcal{E}}_u} 

\newcommand{\Ef}{\mathfrak{E}}
\newcommand{\Df}{\mathfrak{D}}

\newcommand{\betaopt}{{\beta}_{\sf opt}}

\newcommand{\wt}{{\sf wt}}
\newcommand{\minrk}{{\rm minrk}}

\title{Locally Decodable Index Codes}
\author{Lakshmi Natarajan, Prasad Krishnan, V.\ Lalitha and Hoang Dau
\thanks{This article was presented in part at the 2018 IEEE International Symposium on Information Theory (ISIT 2018), Vail, Colorado, USA, and in part at the 2019 IEEE International Symposium on Information Theory (ISIT 2019), Paris, France.}%
\thanks{Lakshmi Natarajan is with the Department of Electrical Engineering, Indian Institute of Technology Hyderabad, email: lakshminatarajan@iith.ac.in.}%
\thanks{Prasad Krishnan and V.\ Lalitha are with the Signal Processing \& Communications Research Center, International Institute of Information Technology Hyderabad, India, email:\{prasad.krishnan,\,lalitha.v\}@iiit.ac.in.}%
\thanks{Hoang Dau is with the School of Science, RMIT University, Australia, email: sonhoang.dau@rmit.edu.au.}%
\thanks{Copyright (c) 2017 IEEE. Personal use of this material is permitted. However, permission to use this material for any other purposes must be obtained from the IEEE by sending a request to pubs-permissions@ieee.org.}%
}

\begin{document}

\maketitle

\begin{abstract}
\boldmath
An index code for broadcast channel with receiver side information is \emph{locally decodable} if each receiver can decode its demand by observing only a subset of the transmitted codeword symbols instead of the entire codeword. Local decodability in index coding is known to reduce receiver complexity, improve user privacy and decrease decoding error probability in wireless fading channels. Conventional index coding solutions assume that the receivers observe the entire codeword, and as a result, for these codes the number of codeword symbols queried by a user per decoded message symbol, which we refer to as \emph{locality}, could be large. In this paper, we pose the index coding problem as that of minimizing the broadcast rate for a given value of locality (or vice versa) and designing codes that achieve the optimal trade-off between locality and rate. We identify the optimal broadcast rate corresponding to the minimum possible value of locality for all single unicast problems. We present new structural properties of index codes which allow us to characterize the optimal trade-off achieved by: vector linear codes when the side information graph is a directed cycle; and scalar linear codes when the minrank of the side information graph is one less than the order of the problem. We also identify the optimal trade-off among all codes, including non-linear codes, when the side information graph is a directed 3-cycle. Finally, we present techniques to design locally decodable index codes for arbitrary single unicast problems and arbitrary values of locality. 
\end{abstract}


\section{Introduction} \label{sec:introduction}

\IEEEPARstart{T}{he} fundamental communication problem in broadcast channels is to design an efficient coding scheme to satisfy the demands of multiple clients with minimal use of the shared communication medium.
When the clients have \emph{side-information} about the messages demanded by other users in the network, \emph{index coding}~\cite{BiK_INFOCOM_98,YBJK_IEEE_IT_11} can achieve remarkable savings in channel use by broadcasting a coded version of the information symbols.
All the receivers then simultaneously decode their demands using the broadcast codeword and their individual side-information.
The objective of index coding is to minimize the number of uses of the broadcast channel, or equivalently, the \emph{broadcast rate}. 
Index coding is a central problem in multi-user communication networks, not only because of its wide ranging applications, such as video-on-demand and daily newspaper delivery~\cite{BiK_INFOCOM_98}, but also because of its strong relation to other coding theoretic problems, such as network coding~\cite{RSG_IEEE_IT_10,EEL_IT_15}, coded caching~\cite{MaN_IT_14}, codes for distributed data storage~\cite{Maz_ISIT_14,ShD_ISIT_14} and distributed computation~\cite{LLPPR_IT_18}.

Of the several classes of index coding problems discussed in the literature since~\cite{BiK_INFOCOM_98}, the most widely studied is \emph{single unicast index coding}, in which each message available at the server is demanded by a unique client.
Several approaches have been taken, most popularly via graph theoretic ideas, to bound the optimal index coding rate from above and below; see, for example, \cite{BiK_INFOCOM_98,YBJK_IEEE_IT_11,ALSWH_FOCS_08,BKL_arxiv_10,CASL_ISIT_11,TDN_ISIT_12,BKL_IT_13,NTZ_IT_13,SDL_ISIT_13,SDL_ISIT_14,ArK_ISIT_14,AgM_arXiv_16}. 
The techniques used in these works naturally lead to converse results and constructions of index codes, most of the constructions being those of (scalar and vector) linear index codes. 

\begin{table*}[!t]
\centering
\renewcommand{\arraystretch}{1.5}
\begin{tabulary}{\linewidth}{|L|L|C|C|}
\Xhline{3\arrayrulewidth}
\textsc{Section}   & \textsc{Summary of Contents} & \textsc{Class of Problems Considered} & \textsc{Main Results} \\
\Xhline{3\arrayrulewidth}
\hline
Section~\ref{sec:system_model} & Definitions and preliminaries & All & \\
\hline
\hline
Section~\ref{sec:min_locality} & Optimal $\beta$ for $r=1$ among all index codes & All & Theorem~\ref{thm:frac_coloring} \\
\hline
\hline
Section~\ref{sec:linear_structure} & Properties of locally decodable vector linear index codes & All & Theorem~\ref{thm:zeroes} \\
\hline
Section~\ref{sec:sub:directed_cycles} & Trade-off between $\beta$ and $r$ among vector linear index codes & Directed cycles & Theorem~\ref{thm:directed_cycles} \\
\hline
\hline
Section~\ref{sec:3cycle} & Trade-off between $\beta$ and $r$ among all index codes & Directed $3$-cycle & Theorem~\ref{thm:3cycle} \\
\hline
\hline
Section~\ref{sec:sub:scalar_structure} & Properties of locally decodable scalar linear index codes & All & Corollary~\ref{cor:lower_bound_r} \\
\hline
Section~\ref{sec:sub:minrank_large} & Trade-off between $\beta$ and $r$ among scalar linear codes & ${\rm minrank}=N-1$ & Theorem~\ref{thm:minrank_N-1} \\
\hline
Section~\ref{sec:sub:receiver_localities_cycles} & Feasible receiver localities among scalar linear codes & Directed cycles & Theorem~\ref{thm:cycle_feasible} \\
\hline
\hline
Section~\ref{sec:design} & Constructions of locally decodable index codes & All & \\
\Xhline{3\arrayrulewidth}
\end{tabulary}
\caption{Summary of the contents and the organization of this paper. Here $\beta$ denotes the index coding broadcast rate, $r$ is the locality and $N$ is the number of messages or receivers in the index coding problem.}
\label{table:summary}
\end{table*}

Conventional index coding solutions in the literature assume that each receiver observes the entire transmitted codeword. 
If the network involves a large number of receivers, the length of the index code could also be large, and thus, the number of transmissions that each receiver has to observe could be significantly larger than the size of its demanded message.
Thus index coding solutions which are optimal in terms of broadcast rate could be unfavorable in certain applications, such as when the power available at the wireless receivers is limited and the receivers can not afford to listen to radio transmissions for an extended period of time. 
In such scenarios it is desirable to use an index code that is \emph{locally decodable}~\cite{HaL_ISIT_12}, i.e., an index code that requires each receiver to observe or query only a part of the transmitted codeword.
We define the \emph{locality} of an index code to be the ratio of the number of codeword symbols observed by a receiver to the number of message symbols it demands.
The objective of designing locally decodable index codes is to minimize the broadcast rate and the locality simultaneously, and achieve the optimal trade-off between these two parameters.

Designing a locally decodable index code is a `client aware' approach to the broadcast problem that takes into account the cost incurred by the clients or the receivers while participating in the communication protocol. 
%
If the index coded symbols are broadcast across different time or frequency resource blocks, a client might want to minimize the number of blocks it must listen to in order to decode its desired index coded message, so as to reduce its power consumption or utilize the remaining blocks to participate in other communication sessions.
Locally decodable index codes allow us to reduce the number of coded packets that must be observed by each client.
In contrast, conventional index coding (without locality considerations) is more `channel centric' with its emphasis purely on minimizing the number of channel uses. 

To the best of our knowledge, the idea of local decodability in index coding was introduced in~\cite{HaL_ISIT_12} where the broadcast rates of random index coding problems, modeled as random graphs, were analyzed under a locality requirement. 
Locally decodable index codes were 
studied under the terminology `$k$-limited access schemes' in~\cite{KSCF_arXiv_18} and were shown to improve user privacy in index coding. 
The authors of~\cite{KSCF_arXiv_18} provide constructions that modify any given binary scalar linear index code into a locally decodable scalar linear code at the cost of increased broadcast rate.
The technique of allowing each receiver to observe only a fraction of the index coded symbols was also used in~\cite{TRAR_TVT_17} and~\cite{VaR_arxiv_19}, to reduce the error rate in wireless fading channels and to reduce the receiver complexity, respectively.

While the notion of locality has been known in the literature, this parameter has not been rigorously formalized before.
Also, a fundamental treatment of locally decodable index codes is not available to the best of our knowledge.
In this work, we present a formal structure to the discussion regarding locally decodable index codes and present constructions of such codes along with their locality parameters for single unicast index coding problems. 
We derive some structural properties of index coding solutions that relate locality to broadcast rate. We use these properties to determine the optimal trade-off between broadcast rate and locality for some families of index coding problems.

\subsection*{Contributions}

We now describe the contributions and organization of this paper. Table~\ref{table:summary} provides a summary of the contributions and the main results of this paper.
In order to help the readers identify these main results in this paper we have highlighted them using boxes.

We formalize the ideas of locality and locality-broadcast rate trade-off in Section~\ref{sec:system_model}, and show that the minimum possible locality for any single unicast index coding problem is one.
In Section~\ref{sec:min_locality}, we identify the optimal broadcast rate corresponding to unit locality for any single unicast problem (considering non-linear index codes as well). This fully characterizes the minimum locality point in the locality-rate trade-off curve.
In Section~\ref{sec:directed_cycles}, we first derive some locality-related properties of vector linear index codes. We then utilize these results to derive the optimal locality-rate trade-off achieved by linear codes for the class of index coding problems whose side information graphs are directed cycles.
For the specific case of directed $3$-cycle we identify the optimal locality-rate trade-off achieved by any index code, including non-linear codes, in Section~\ref{sec:3cycle}.
In Section~\ref{sec:minrank_N_minus_1}, we derive some structural properties of scalar linear locally decodable index codes. We then characterize the locality-rate trade-off achieved by scalar linear codes for the family of problems whose minrank is one less than the number of messages. 
Note that directed cycles considered in Section~\ref{sec:directed_cycles} are a subset of the family of side information graphs considered in Section~\ref{sec:minrank_N_minus_1}.
In Section~\ref{sec:minrank_N_minus_1}, we also characterize the set of all possible receiver localities for directed cycles that can be achieved using scalar linear index codes.
Sections~\ref{sec:min_locality},~\ref{sec:directed_cycles},~\ref{sec:3cycle}, and~\ref{sec:minrank_N_minus_1} consider specific families of index coding problems or a specific value of locality, and identify the optimal solutions for these cases.
In Section~\ref{sec:design} we present techniques to design locally decodable index codes (which are not necessarily optimal) for arbitrary single unicast problems and arbitrary values of locality. 
We also show how the traditional \emph{partition multicast}~\cite{BiK_INFOCOM_98,TDN_ISIT_12} and \emph{cycle covering}~\cite{NTZ_IT_13,CASL_ISIT_11} solutions to index coding can be modified to yield locally decodable index codes.
We conclude the paper in Section~\ref{sec:conclusion} and discuss the relation between locally decodable index codes and overcomplete dictionaries for sparse representation of vectors.

The optimal index coding schemes identified in Sections~\ref{sec:min_locality},~\ref{sec:directed_cycles},~\ref{sec:3cycle}, and~\ref{sec:minrank_N_minus_1} are based on known index coding techniques. 
However, one of the main contributions of this paper is the development of new tools to derive good lower bounds on rate and locality that are vital in proving the optimality of these schemes.


\emph{Notation:} For any positive integer $N$, we will denote the set $\{1,\dots,N\}$ by $[N]$. Matrices and column vectors are denoted by bold upper and lower case letters, respectively, such as $\pmb{A}$ and $\pmb{x}$. 
The symbol $\wt$ denotes the Hamming weight of a vector.
The finite field of size $q$ is denoted as $\Fb_q$. The subspace spanned by vectors $\pmb{u}_1,\dots,\pmb{u}_N$ is denoted by $\spp(\pmb{u}_1,\dots,\pmb{u}_N)$. 
The column space of a matrix $\pmb{A}$ is denoted as $\colsp(A)$ and the null space of $\pmb{A}$ is $\nullsp(\pmb{A})=\{\pmb{x}|\pmb{Ax}=\pmb{0}\}$.
The support set of a vector $\pmb{x}$ is denoted as $\supp(\pmb{x})$.
For a set of column vectors $\xb_1,\dots,\xb_N$ and a subset $K \subset [N]$, we define $\xb_K$ to be the column vector obtained by concatenating $\xb_j, j \in K$. The empty set is denoted by $\phi$. For a matrix $\p{A}$, the component in $j^{\text{th}}$ row and $i^{\text{th}}$ column is denoted as $\p{A}_{j,i}$. The transpose of $\p{A}$ is denoted as $\p{A}^\tr$.

\section{System Model \& Preliminaries} \label{sec:system_model}

We consider single unicast index coding for a broadcast channel consisting of $N$ receivers $\rx_1,\dots,\rx_N$. The transmitter holds $N$ messages $\pmb{x}_1,\dots,\pmb{x}_N$ where the $i^\tth$ message $\xb_i$ is demanded by $\rx_i$, and the messages $\pmb{x}_j$, $j \in K_i$, are known at this receiver as side information, where $K_i \subseteq [N] \setminus \{i\}$. 
The side information graph $G=(\mathcal{V},\mathcal{E})$ is the directed graph with vertex set $\mathcal{V} = [N]$ and edge set $\mathcal{E}=\{(i,j)\,|\, i \in [N], j \in K_i\}$.
The transmitter encodes the messages into a codeword $\cb$ and broadcasts the codeword to all the receivers through a noiseless broadcast channel.
Each receiver decodes its demands using $\cb$ and its own side information.
Throughout this paper we will consider only single unicast index coding problems and denote a problem by its side information graph $G$.



We assume that the messages $\xb_1,\dots,\xb_N$ are vectors of length $M$ over a finite alphabet $\mathcal{A}$, with $|\mathcal{A}| > 1$, and the codeword $\cb$ is a vector of length $\ell$ over the same alphabet, i.e.,
$\xb_i \in \mathcal{A}^M$, $i \in [N]$, and $\cb \in \mathcal{A}^\ell$.
We will assume that the alphabet $\mathcal{A}$ is arbitrary but fixed for a given index coding problem. 
The {\em broadcast rate} $\beta=\ell/M$ of an index code is the ratio of the codeword length to message length, and measures the bandwidth or time required by the source to broadcast the coded symbols to all the receivers.
Please note that our formulation includes non-linear index codes also. 

\subsection{Locally Decodable Index Codes}

Unlike the conventional index coding problem where each receiver is required to observe or query the entire codeword $\pmb{c}$, we allow the receivers to observe only a part of the codeword in order to decode their demands.
Index codes that satisfy this property are called \emph{locally decodable}~\cite{HaL_ISIT_12}.
We will assume that $\rx_i$ queries the subvector $\pmb{c}_{R_i}=(c_j, j \in R_i)$, where $R_i \subseteq [\ell]$ is chosen in such a way that $\rx_i$ can decode $\pmb{x}_i$ using $\pmb{c}_{R_i}$ and the available side information $\pmb{x}_j$, $j \in K_i$.
The \emph{locality of the receiver} of $\rx_i$ is $r_i = \frac{|R_i|}{M}$ is the ratio of the number of codeword symbols observed by $\rx_i$ to the number of information symbols demanded.
\begin{definition}
The {\em overall locality} or simply the {\em locality} of an index code is defined as 
\begin{equation} \label{eq:defn_r}
r = \max_{i \in [N]} \, r_i = \max_{i \in [N]} \, \frac{|R_i|}{M}
\end{equation} 
and is equal to the maximum number of coded symbols queried by any of the $N$ receivers per decoded information symbol. The \emph{average locality} of an index code is  
\begin{equation} \label{eq:ravg}
\ravg = \sum_{i \in [N]}\frac{r_i}{N} = \sum_{i \in [N]} \frac{|R_i|}{MN}.
\end{equation} 
\end{definition}

Observe that the average locality is upper bounded by overall locality $\ravg \leq r$.
Without loss of generality we will consider only those index codes for which $R_1 \cup \cdots \cup R_N = [\ell]$ since the subset $\{c_j~|~ j \in [\ell] \setminus (R_1 \cup \cdots \cup R_N) \}$ of codeword symbols 
need not be generated or transmitted by the encoder.
For a given locality $r$, it is desirable to use an index code with as small a value of $\beta$ as possible and vice versa, which leads us to the following definition.

\begin{definition}
Given an index coding problem $G$, the {\em optimal broadcast rate function} $\beta_{G}^*(r)$ is the infimum of the broadcast rates among all message lengths \mbox{$M \geq 1$} and all valid index codes for $G$ with locality at the most $r$. 
\end{definition}

The function $\beta_G^*(r)$ captures the trade-off between broadcast rate and locality, i.e., between the reduction in the number of channel uses possible through coding and the number of codeword symbols that a receiver has to observe to decode each message symbol. 

\begin{remark}
Our system model assumes that the codeword vector and the message vectors are defined over the same alphabet $\Ac$.
In some scenarios the underlying alphabet, say $\Zc$, for the codeword $\p{c}$ and the alphabet $\Ac$ for the messages $\p{x}_1,\dots,\p{x}_N$ can be different.
In Appendix~\ref{app:invariance_alphabet} we generalize the definition of locality to the case where $\Zc$ can be different from $\Ac$, and show that the optimal broadcast rate function is independent of the choice of $\Zc$ and $\Ac$ (see Lemmas~\ref{lem:app:channel_alphabet} and~\ref{lem:app:message_alphabet} in Appendix~\ref{app:invariance_alphabet}). 
Hence, without loss of generality, we will only consider the case $\Zc=\Ac$ in the rest of this paper. Also, from Lemma~\ref{lem:app:message_alphabet} of Appendix~\ref{app:invariance_alphabet}, the value of $\beta_G^*(r)$ is independent of the choice of $\Ac$.  
\end{remark}

\begin{remark}
The cost incurred by a receiver in participating in an index coded broadcast can also be measured as the fraction of codeword symbols queried by the receiver. 
The fraction of coded symbols observed by $\rx_i$ is $\delta_i = |R_i|/\ell$, and the worst-case among all receivers is $\delta = \max_i \delta_i$. Note that $\delta_i = r_i/\beta$, and hence, $\delta=r/\beta$.
The design problem posed in this paper is to construct index codes that minimize $\beta$ for a given value of $r$, or equivalently, minimize $r$ for a given value of $\beta$.
Since $\delta=r/\beta$, we observe that for a given value of $\beta$, minimizing $\delta$ is the same problem as minimizing $r$.
Thus both choices of performance parameters, $r$ and $\delta$, capture the same underlying index coding problem. 
\end{remark}

We now prove some elementary properties of the function $\beta_G^*(r)$.
We will rely on information-theoretic inequalities for this purpose. We will assume that the messages $\xb_1,\dots,\xb_N$ are random, independent of each other and are uniformly distributed in $\mathcal{A}^M$. The logarithms used in measuring mutual information and entropy will be calculated to the base $|\mathcal{A}|$. For example, the entropy of $\xb_i$ is $H(\xb_i)=M$ since $\xb_i$ is uniformly distributed in $\mathcal{A}^M$. 

\begin{lemma}
The locality $r$ and the average locality $\ravg$ of any valid index coding scheme satisfy \mbox{$r, \ravg \geq 1$}. 
\end{lemma}
\begin{proof}
Considering the decoder at $\rx_i$, we have $I(\xb_i;\cb_{R_i},\xb_{K_i})=H(\xb_i)=M$. 
Since $\xb_i$ is independent of $\xb_{K_i}$ (because $i \notin K_i$), $I(\xb_i;\xb_K)=0$, and
\begin{align}
M &= I(\xb_i;\cb_{R_i},\xb_{K_i}) = I(\xb_i;\xb_{K_i}) + I(\xb_i;\cb_{R_i}|\xb_{K_i}) \nonumber \\
  &= I(\xb_i;\cb_{R_i}|\xb_{K_i}) \leq H(\cb_{R_i}) \leq |R_i|. \label{eq:simple_bound_on_r}
\end{align} 
Hence, $r_i = |R_i|/M \geq 1$ for all $i \in [N]$, and thus, $r, \ravg \geq 1$. 
\end{proof}

Note that uncoded transmission, i.e., $\cb=(\xb_1,\dots,\xb_N)$, is a valid index code with $r=1$. Hence, we will assume that the domain of the function $\beta_G^*$ is the interval $1 \leq r < \infty$. 

\begin{lemma} \label{lem:convexity}
The function $\beta_G^*(r)$ is convex and non-increasing. 
\end{lemma}
\begin{proof}
The non-increasing property of $\beta_G^*$ follows immediately from its definition. 
We use time-sharing to prove convexity, see Appendix~\ref{app:lem:convexity} for details.
\end{proof}

For any valid index code, we have $|R_i| \leq \ell$, and hence, $r = \max_i |R_i|/M \leq \ell/M = \beta$. Hence, if there exists a valid index code with broadcast rate $\beta$, then its locality is at the most $\beta$, and hence, 
\begin{equation} \label{eq:betaG_simple_1}
\beta_G^*(\beta) \leq \beta. 
\end{equation} 
We will denote by $\betaopt(G)$ the infimum among the broadcast rates of all valid index codes for $G$ (considering all possible message lengths $M \geq 1$ and all possible localities $r \geq 1$). Then it follows that $\beta_G^*(r) \geq \betaopt$ for all $r \geq 1$. Together with~\eqref{eq:betaG_simple_1}, choosing $\beta=\betaopt$, we deduce
\begin{equation} \label{eq:betaG_simple_2}
\beta_G^*(\betaopt) = \betaopt.
\end{equation} 

\begin{example}[\emph{Locality-broadcast rate trade-off of the directed $3$-cycle}] \label{ex:first}
\begin{figure}[!t]
\centering
\includegraphics[width=2in]{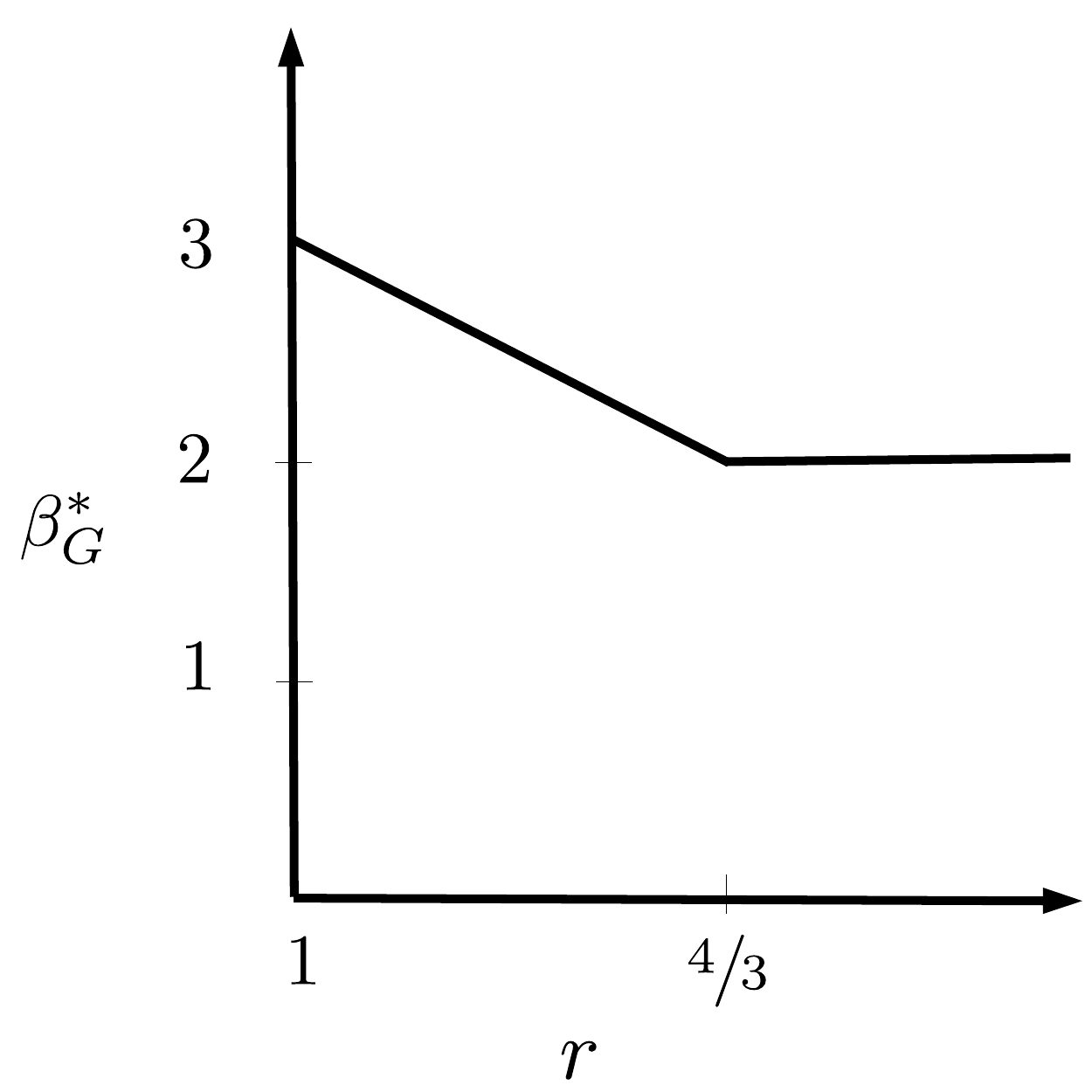}
\caption{The trade-off between the broadcast rate $\beta_G^*(r)$ and the locality $r$ for the index coding problem represented by the $3$-cycle $G$.}
\label{fig:3cycle_tradeoff}
\end{figure} 
Let $G$ be the directed $3$-cycle, i.e., $N=3$ and the side information at the three receivers are $K_1=\{2\}$, $K_2=\{3\}$ and $K_3=\{1\}$.
We show in Section~\ref{sec:3cycle} that for this index coding problem
\begin{equation*} 
 \beta_G^*(r) = \max\{6-3r,2\} \text{ for all } r \geq 1.
\end{equation*} 
This function is shown in Fig.~\ref{fig:3cycle_tradeoff}. We observe that in order to achieve any savings in rate compared to the uncoded transmission (\mbox{$\beta=N=3$}), we necessarily require the locality to be strictly greater than $1$, i.e., each receiver must necessarily query more codeword symbols than the message length to achieve savings in the broadcast channel uses. Also, the smallest locality required to achieve the minimum rate \mbox{$\betaopt=2$} is \mbox{$r=4/3$}.
\end{example}

\subsection{Locally Decodable Linear Index Codes}

Some of the results presented in this paper are applicable to the class of linear index codes only.
Linear index codes form the most widely studied family of index coding solutions, and enjoy low complexity encoding and decoding operations.
When considering (vector) linear index codes, we will assume that each message $\pmb{x}_i$ is a vector of length $M$ over some finite field $\Fb_q$, where $q$ denotes the size of the field. 
The $M$ components of the $i^\tth$ message vector $\pmb{x}_i$ are denoted as $x_{i,1},\dots,x_{i,M}$.
Encoding is performed by first concatenating the $N$ messages into $\pmb{x} = (\pmb{x}_1^\tr,\dots,\pmb{x}_N^\tr)^\tr \in \Fb_q^{MN}$ and multiplying this vector with a carefully designed encoding matrix $\pmb{L} \in \Fb_q^{MN \times \ell}$ to generate a length $\ell$ the codeword $\pmb{c}^\tr = \pmb{x}^\tr \pmb{L}$.
Note that the $MN$ components of the concatenated vector $\pmb{x}=(x_1,\dots,x_{MN})^\tr$ and the components of the individual message vectors are related as $x_{(i-1)M+m} = x_{i,m}$ for $i \in [N]$ and $m \in [M]$.

In the sequel we would like to view a vector linear index coding problem involving $N$ vector messages $\pmb{x}_1,\dots,\pmb{x}_N \in \Fb_q^M$ as a scalar linear problem defined over $MN$ scalar messages $x_1,\dots,x_{MN} \in \Fb_q$. In this case the $i^\tth$  receiver demands $M$ scalar messages $x_{(i-1)M+1},x_{(i-1)M+2},\dots,x_{iM}$, which correspond to the $M$ components of the vector $\pmb{x}_i$. This set of demands of $\rx_i$ is represented by the index set $\Db_i = \{(i-1)M+m \,|\, m \in [M] \}$. The scalar symbols available as side information at $\rx_i$ correspond to the index set 
\begin{equation*}
 \Kb_i = \left\{\, (j-1)M + m \, ~\big\vert~ \, j \in K_i, \, m \in [M] \,\right\}.
\end{equation*} 
Thus the vector linear problem corresponding to the side information graph $G$ with message length $M$ is equivalent to a scalar linear problem with $N$ receivers and $MN$ messages, where the index set of demands of $\rx_i$ is $\Db_i$ and the set corresponding to the side information at $\rx_i$ is $\Kb_i$. 
Note that $\Db_i$ and $\Kb_i$ are subsets of $[MN]$.

We would like to characterize the trade-off between the broadcast rate $\beta$ and the locality $r$ of \emph{linear} index codes over $\Fb_q$ for a given index coding problem $G$. We define the optimum locality-rate trade-off among all linear index codes for $G$ over $\Fb_q$ as
\begin{align*}
 \beta^*_{G,q}(r) = \inf \big\{ \beta ~\, | ~\, \exists &\text{ a linear index code over } \Fb_q \\
 &\text{ for } G \text { with rate } \beta \text{ and locality} \leq r \,\big\},
\end{align*} 
where the infimum considers linear index codes over all possible message lengths $M \geq 1$ for the index coding problem $G$. 
Note the difference between $\beta^*_{G,q}(r)$ and $\beta^*_{G}(r)$; while the former is the optimal trade-off between rate and locality among linear index codes (over $\Fb_q$), the latter considers all valid index codes (including linear and non-linear codes) for $G$.

\begin{example}[\emph{A simple scalar linear code for directed cycles}] \label{ex:simple_code_cycle}
Let the message length $M=1$, $\Fb_q$ be any finite field, and $G$ be the directed cycle of length $N$, i.e., $K_i=\{i+1\}$ for $i \in [N-1]$ and $K_N=\{1\}$. Consider the $N \times (N-1)$ encoder matrix
\begin{equation*}
\pmb{L} = \begin{bmatrix} 
          1 & 1 & \cdots & 1 \\
          1 & 0 & \cdots & 0 \\
          0 & 1 & \cdots & 0 \\
          \vdots &  &  & \vdots \\
          0 & 0 & \cdots & 1
          \end{bmatrix},
\end{equation*} 
that generates the codeword 
\begin{equation*}
\pmb{c}^\tr \!=\! (c_1,\dots,c_{N-1}) \!=\! \pmb{x}^\tr\pmb{L} \!=\! (x_1+x_2,x_1+x_3,\dots,x_1+x_N).
\end{equation*} 
The receivers $\rx_1$ and $\rx_N$ can decode their demands by querying $c_1$ and $c_{N-1}$, respectively, and hence, $|R_1|=|R_N|=1$. For \mbox{$1 < i < N$}, receiver $\rx_i$ queries the symbols \mbox{$c_{i-1}=x_1+x_i$} and \mbox{$c_{i}=x_1 + x_{i+1}$}, and uses its side information $x_{i+1}$ to compute \mbox{$c_{i-1} - c_i - x_{i+1}$}, which equals its demand $x_i$. Hence, $|R_i| = 2$ for $1 < i < N$. Since $M=1$, the locality of each receiver $r_i=1$ if $i=1$ or $N$, and $r_i=2$ otherwise. The overall locality $r=2$ and the average locality $\ravg=2(N-1)/N$.

Since the graph $G$ is symmetric, for any choice of $i \in [N-1]$, the above coding scheme can be modified by an appropriate permutation of the rows of $\pmb{L}$ to allow receiver localities $r_i=1$ and $r_{i+1}=1$ and locality $r_j=2$ at all other receivers $j \neq i,i+1$. 
\end{example}

\section{Minimum-Locality Point\\ in the Locality-Rate Trade-off} \label{sec:min_locality}

The minimum possible locality for any valid index coding solution is \mbox{$r=1$}. Note that $r=1$ immediately implies $\ravg=1$ as well. 
In this section we will show that the optimal index coding scheme with locality $r=1$, considering non-linear schemes also, is the index code based on the fractional coloring of the \emph{interference graph} corresponding to the problem.
This explicitly identifies the value of the locality-rate trade-off function $\beta^*_G$ at $r=1$.
In~\cite{HaL_ISIT_12} it is remarked that the optimal broadcast rate with $r=1$ among scalar linear index codes (i.e., $M=1$, $\mathcal{A}$ is a finite field and the encoder is a linear transformation) is the chromatic number of the interference graph. 
Our results in this section generalize the result of~\cite{HaL_ISIT_12} by considering all valid index codes, including non-linear codes.
We first recall some graph-theoretic terminology related to index coding as used in~\cite{BKL_arxiv_10,SDL_ISIT_13}, state the main result of this section as Theorem~\ref{thm:frac_coloring} below and then prove the achievability and converse parts of this theorem in Sections~\ref{sec:r_1_achieve} and~\ref{sec:r_1_converse}, respectively.

The {\em underlying undirected side information graph} \mbox{$G_u=(\Vc,\Ec_u)$} of the side information graph $G=(\Vc,\Ec)$ is the graph with vertex set \mbox{$\Vc=[N]$} and an undirected edge set 
\begin{equation*}
\Ec_u=\left\{\,\{i,j\} \, | \, (i,j),\,(j,i) \in \Ec \right\}, 
\end{equation*} 
that is, $\{i,j\} \in \Ec_u$ if and only if \mbox{$i \in K_j$} and \mbox{$j \in K_i$}. 
The {\em interference graph} $\Gcpu=(\Vc,\Ecpu)$ is the undirected complement of the graph $G_u$, i.e., $\Ecpu=\left\{\, \{i,j\} \, | \, \{i,j\} \notin \Ec_u \right\}$. Note that 
\begin{equation} \label{eq:Ecpu}
\{i,j\} \in \Ecpu \text{ if and only if either } i \notin K_j \text{ or } j \notin K_i.
\end{equation} 

For positive integers $a$ and $b$, an \mbox{\em $a:b$ coloring} of the undirected graph $\Gcpu=(\Vc,\Ecpu)$ is a set $\{C_1,C_2,\dots,C_N\}$ of $N$ subsets $C_1,\dots,C_N \subseteq [a]$, such that $|C_1|=\cdots=|C_N|=b$ and $C_i \cap C_j = \phi$ if $\{i,j\} \in \Ecpu$. The elements of $[a]$ are {\em colors}, and each vertex of $\Gcpu$ is assigned $b$ colors such that no two adjacent vertices have any colors in common.
The {\em fractional chromatic number} $\chi_f$ of the undirected graph $\Gcpu$ is
\begin{equation*}
 \chi_f(\Gcpu) = \min \left\{\, \frac{a}{b} ~\Big\vert~ \text{an } a:b \text{ coloring of } \Gcpu \text{ exists} \right\}.
\end{equation*} 
The fractional chromatic number is a rational number and can be obtained as a solution to a linear program~\cite{godsil2013algebraic}. The \emph{chromatic number} $\chi(\Gcpu)$ of the graph $\Gcpu$ is the smallest integer $a$ such that an $a:1$ coloring of $\Gcpu$ exists. In general, we have $\chi_f(\Gcpu) \leq \chi(\Gcpu)$.
The main result of this section is

\begin{tcolorbox}[colback=white,boxsep=3pt,left=2pt,right=2pt,top=2pt,bottom=2pt]
\begin{theorem} \label{thm:frac_coloring}
For any single unicast index coding problem $G$, the optimal broadcast rate for locality $r=1$ is $\beta_G^*(1) = \chi_f(\Gcpu)$.
\end{theorem}
\end{tcolorbox}

\subsection{Proof of Theorem~\ref{thm:frac_coloring}: Achievability Part} \label{sec:r_1_achieve}

It is well known that there exists a coding scheme, which is known as the {\em fractional clique covering}  or the {\em fractional coloring} solution, for any index coding problem $G$ with broadcast rate $\beta=\chi_f(\Gcpu)$, for example, see~\cite{BKL_arxiv_10}. 
It is straightforward to observe that $r=1$ for this coding scheme.
It then follows that $\beta_G^*(1) \leq \chi_f(\Gcpu)$. 
For completeness, we now recall the fractional coloring solution and deduce that $r=1$ for this scheme.

Let $\chi_f(\Gcpu)=a/b$ for integers $a$ and $b$, and let $C_1,\dots,C_N \subseteq \{1,\dots,a\}$ be an $a:b$ coloring of $\Gcpu$. Set the codeword length \mbox{$\ell=a$} and message length \mbox{$M=b$}. Denote the components of the message vectors $\xb_i \in \Ac^M$ using the variables $w_{i,t} \in \Ac$ as follows: 
$\xb_i = \left( w_{i,t}, t \in C_i \right)$, i.e., one message symbol $w_{i,t}$ corresponding to each color $t$ in the set $C_i$.
Endow the set $\Ac$ with any abelian group structure $(\Ac,+)$.
The symbols of the codeword $\cb=(c_1,\dots,c_\ell)$ are generated as 
$c_t = \sum_{i:\, t \in C_i} w_{i,t}$, for $t \in [\ell]$.
%
Decoding at $\rx_i$ is performed as follows. Note that $\xb_i$ is composed of all symbols $w_{i,t}$ such that $t \in C_i$. In order to decode $w_{i,t}$, the receiver retrieves the code symbol $c_t$ which is related to $w_{i,t}$ as 
\begin{equation} \label{eq:frac_color_ach}
 c_t = w_{i,t} + \sum_{\substack{j \neq i \\ j: \, t \in C_j}} w_{j,t}.
\end{equation} 
For any choice of the index $j$ in the summation above, we have $i \neq j$ and $t \in C_i \cap C_j$. Since $C_1,\dots,C_N$ is a coloring of $\Gcpu$ and $C_i \cap C_j \neq \phi$, we deduce that $\{i,j\} \notin \Ecpu$, or equivalently, $i \in K_j$ and $j \in K_i$. Hence, for each $j \neq i$ such that $t \in C_j$, $\rx_i$ knows the value of $w_{j,t}$, and thus, can recover $w_{i,t}$ from $c_t$ using~\eqref{eq:frac_color_ach}. Using a similar procedure $\rx_i$ can decode all the $M$ symbols in $\xb_i$ from the $M$ coded symbols $(c_t,t \in C_i)$. This decoding method uses $R_i=C_i$, and hence, $r_i=|R_i|/M=1$ for all $i \in [N]$, implying $r=1$.

\subsection{Proof of Theorem~\ref{thm:frac_coloring}: Converse Part} \label{sec:r_1_converse}

The key result in the proof of the converse is Lemma~\ref{lem:r_1_R_non_intersecting}, which states that $R_i \cap R_j = \phi$ for all $i,j$ such that $\{i,j\} \in \Ecpu$. Before providing the proof of the converse, we will give the intuition behind Lemma~\ref{lem:r_1_R_non_intersecting}. The idea behind Lemma~\ref{lem:r_1_R_non_intersecting} is based on the analysis of scalar linear codes over $\Fb_2$ with locality $1$; see also~\cite{HaL_ISIT_12}. 
Note that each codeword symbol $c_k$ of a binary scalar linear index code is a sum of a subset $S_k$ of the scalars $x_1,\dots,x_N \in \Fb_2$. Since the locality is $1$, each receiver $\rx_i$ observes exactly one codeword symbol, say $c_{f(i)}$. This implies that the set $S_{f(i)}$ must contain $x_i$, and any other symbol present in $S_{f(i)}$ must be available as side information at $\rx_i$. Now, if $j \notin K_i$, i.e., $\{i,j\} \in \Ecpu$, then necessarily $x_j \notin S_{f(i)}$. On the contrary, the code symbol $c_{f(j)}$ observed by $\rx_j$ will satisfy the property that $x_j \in S_{f(j)}$. We conclude that the code symbols $c_{f(i)}$ and $c_{f(j)}$ observed by $\rx_i$ and $\rx_j$, respectively, are necessarily distinct. Thus $R_i \cap R_j =\{f(i)\} \cap \{f(j)\} =  \phi$.
Lemma~\ref{lem:r_1_R_non_intersecting} extends this result to all general index codes with locality $1$.

\emph{Proof of the converse:} Consider any valid index code, possibly non-linear, for $G$ with locality \mbox{$r=1$}, message length $M$ and code length $\ell$.
We assume that $\xb_1,\dots,\xb_N$ are independent and uniformly distributed over $\Ac^M$. Let $\cb$ be the codeword generated by the index code. 
%
Since $r=1$, from~\eqref{eq:defn_r} 
and~\eqref{eq:simple_bound_on_r}, we deduce that $|R_1|=\cdots=|R_N|=M$ for this valid index code.

\begin{lemma} \label{lem:basic_lemma_1}
For any $i \in [N]$ and any $P \subset [N]$ such that $i \notin P$, we have 
{(i)} $I(\cb_{R_i};\xb_i|\xb_{K_i \cup P})=M$; 
{(ii)} $H(\cb_{R_i}|\xb_i,\xb_{K_i})=0$; and
{(iii)} $H(\cb_{R_i}|\xb_P)=M$.

\end{lemma}
\begin{proof}
We first observe that
$I(\cb_{R_i};\xb_i|\xb_{K_i \cup P})=H(\xb_i|\xb_{K_i \cup P}) - H(\xb_i|\cb_{R_i},\xb_{K_i \cup P})$. 
Since \mbox{$i \notin K_i \cup P$}, $\xb_i$ is independent of $\xb_{K_i \cup P}$. Also, $\xb_i$ can be decoded using $\cb_{R_i}$ and $\xb_{K_i}$. Hence, $H(\xb_i|\xb_{K_i \cup P})=M$ and $H(\xb_i|\cb_{R_i},\xb_{K_i \cup P})=0$, thereby proving part~\emph{(i)}.

Using the result in part~\emph{(i)} and decomposing the mutual information term $I(\cb_{R_i};\xb_i|\xb_{K_i \cup P})$, we have
\begin{align} \label{eq:key_for_basic_lemma_0}
M = H(\cb_{R_i}|\xb_{K_i \cup P}) - H(\cb_{R_i}|\xb_i,\xb_{K_i \cup P}). 
\end{align} 
Since $\cb_{R_i}$ is a length $M$ vector, we have $H(\cb_{R_i}|\xb_{K_i \cup P}) \leq M$. 
Also, $H(\cb_{R_i}|\xb_i,\xb_{K_i \cup P}) \geq 0$. 
Considering these facts together with~\eqref{eq:key_for_basic_lemma_0}, we deduce that
\begin{equation} \label{eq:key_for_basic_lemma_1}
H(\cb_{R_i}|\xb_i,\xb_{K_i \cup P})=0 \text{ and } H(\cb_{R_i}|\xb_{K_i \cup P}) = M.
\end{equation} 
Observe that~\eqref{eq:key_for_basic_lemma_1} holds for any choice of $P$ such that $i \notin P$. Choosing \mbox{$P=\phi$} in the first equality in~\eqref{eq:key_for_basic_lemma_1} proves part~\emph{(ii)} of this lemma.
%
Now using the fact that $\cb_{R_i}$ is of length $M$, and the second equality in~\eqref{eq:key_for_basic_lemma_1}, we have
\begin{align*}
 M \geq H(\cb_{R_i}) \geq H(\cb_{R_i}|\xb_P) \geq H(\cb_{R_i}|\xb_{K_i \cup P}) = M.
\end{align*} 
This shows that $H(\cb_{R_i}|\xb_P)=M$, proving part~\emph{(iii)}.
\end{proof}

\begin{lemma} \label{lem:r_1_R_non_intersecting}
For any $\{i,j\} \in \Ecpu$, we have $R_i \cap R_j = \phi$.
\end{lemma}
\begin{proof}
Using~\eqref{eq:Ecpu}, we will assume without loss of generality that $j \notin K_i$. We will now assume that $R_i \cap R_j \neq \phi$ and prove the lemma by contradiction. Let \mbox{$t \in R_i \cap R_j$} and \mbox{$P=\{i\} \cup K_i$}. 
From part~\emph{(ii)} of Lemma~\ref{lem:basic_lemma_1}, we have $H(\cb_{R_i}|\xb_i,\xb_{K_i})=0$. In particular, since $t \in R_i$, we have 
\begin{equation} \label{eq:lem:contradiction_1}
H(c_t|\xb_i,\xb_{K_i})=H(c_t|\xb_P)=0.
\end{equation} 

Note that \mbox{$j \notin P$} since \mbox{$j \neq i$} and \mbox{$j \notin K_i$}.
From part~\emph{(iii)} of Lemma~\ref{lem:basic_lemma_1}, we observe that $H(\cb_{R_j}|\xb_P)=M$. This implies that for any given realization of $\xb_P$, the vector $\cb_{R_j}$ is uniformly distributed over $\Ac^M$. Hence, the $M$ coordinates of $\cb_{R_j}$ are independent and uniformly distributed over $\Ac$. Since $t \in R_j$, we conclude that for any given realization of $\xb_P$, $c_t$ is uniformly distributed over $\Ac$, and hence, \mbox{$H(c_t|\xb_P)=1$} which contradicts~\eqref{eq:lem:contradiction_1}. 
\end{proof}

\begin{lemma} \label{lem:r_1_converse}
For any valid index coding scheme for $G$ with $r=1$, the broadcast rate $\beta \geq \chi_f(\Gcpu)$.
\end{lemma}
\begin{proof}
From Lemma~\ref{lem:r_1_R_non_intersecting}, the subsets $R_1,\dots,R_N \subseteq [\ell]$ are such that $R_i \cap R_j =\phi$ if $\{i,j\} \in \Ecpu$ and $|R_i|=M$ for all $i \in [N]$. Hence, $\{R_1,\dots,R_N\}$ is an \mbox{$\ell:M$} coloring of $\Gcpu$. Consequently, the broadcast rate
$\beta = {\ell}/{M} \geq \chi_f(\Gcpu)$.
\end{proof}

Combining the converse result in Lemma~\ref{lem:r_1_converse} with the achievability result in Section~\ref{sec:r_1_achieve}, we arrive at Theorem~\ref{thm:frac_coloring}.
%
%
Theorem~\ref{thm:frac_coloring} can be easily generalized to the case where the message length $M$ is fixed.

\begin{corollary}
The optimal broadcast rate for index coding problem $G$ with locality $r=1$ and message length $M$ is 
\begin{equation*}
\min \left\{ \frac{a}{M} \, \Big\vert \, \text{an } a:M \text{ coloring of } \Gcpu \text{ exists} \,\right\}.
\end{equation*}
In particular, the optimum rate for \mbox{$M=r=1$} is $\chi(\Gcpu)$. 
\end{corollary}
\begin{proof}
The achievability result is similar to the arguments used in Section~\ref{sec:r_1_achieve} with the additional restriction that the subsets of colors $C_1,\dots,C_N$ are all of size $M$. Converse follows by recognizing that the set of subsets $\{R_1,\dots,R_N\}$ is an $\ell:M$ coloring of $\Gcpu$.
When $M=1$, the smallest $\ell$ for which an $\ell:1$ coloring of $\Gcpu$ exists is the chromatic number of $\Gcpu$.
\end{proof}


\section{Vector Linear Codes for Directed Cycles} \label{sec:directed_cycles}

In this section, we will first derive some properties of locally decodable (vector) linear index codes (codes for any message length $M \geq 1$) in Section~\ref{sec:linear_structure} for arbitrary index coding problems.
%
These structural properties of linear index codes will be useful in deriving the rate-locality trade-off of directed cycles presented in Section~\ref{sec:sub:directed_cycles}. 


\subsection{Structure of locally decodable linear index codes} \label{sec:linear_structure}

Following the notation from~\cite{DSC_IT_12}, for a vector $\pmb{u} \in \Fb_q^{MN}$ and set \mbox{$E \subset [MN]$}, we write $\pmb{u} \lhd E$ to denote \mbox{$\supp(\pmb{u}) \subseteq E$}.
Let us denote the columns of the encoder matrix $\pmb{L}$ as $\pmb{L}_1,\dots,\pmb{L}_\ell \in \Fb_q^{MN}$. Then the $k^\tth$ symbol of the codeword is $c_k = \pmb{x}^\tr\pmb{L}_k$.
Note that the $i^\tth$ receiver queries the subvector \mbox{$\pmb{c}_{R_i}=(c_k,k \in R_i)$}, and utilizes the side information $\pmb{x}_{\Kb_i}=(x_j, j \in \Kb_i)$, to decode the demand $\pmb{x}_{\Db_i}=(x_j, j \in \Db_i)$. 

We first consider a necessary and sufficient condition for a given encoder matrix $\pmb{L}$ and receiver queries $R_1,\dots,R_N$ to represent a valid index code for a single unicast problem $G$. 
We observe that the proofs of Lemmas~3.1 and~4.3 and Corollary~4.4 of~\cite{DSC_IT_12} can be directly adapted to the scenario of locally decodable index codes, immediately yielding the following constraints on the encoder matrix and receiver queries.
Let $\pmb{e}_1,\dots,\pmb{e}_{MN}$ be the standard basis of $\Fb_q^{MN}$.

\begin{theorem} \label{thm:design_criterion}
For each $i \in [N]$, let $\rx_i$ query the subvector $\pmb{c}_{R_i}$ of the codeword $\pmb{c}^\tr=\pmb{x}^\tr\pmb{L}$ and have the side information $\pmb{x}_{\Kb_i}$. Then $\rx_i$ can decode its demand $\pmb{x}_{\Db_i}$ if and only if for each $j \in \Db_i$ there exists a $\pmb{u}_j \in \Fb_q^{MN}$ such that $\pmb{u}_j \lhd \Kb_i$ and $\pmb{u}_j + \pmb{e}_j \in \spp(\pmb{L}_k, k \in R_i)$.
\end{theorem}

We will say that $\pmb{L} \in \Fb_q^{MN \times \ell}$ is a valid encoder matrix corresponding to the queries $R_1,\dots,R_N \subseteq [\ell]$ if it satisfies the criterion stated in Theorem~\ref{thm:design_criterion} for decodability at all the receivers.

Observe that among the component symbols in the codeword $\pmb{c}=(c_1,\dots,c_\ell)^\tr$, some are queried exactly once, i.e., queried by a single receiver, and the other symbols are queried by multiple receivers in the network. Let 
$\Sc_i = R_i \setminus \left(R_1 \cup \cdots \cup R_{i-1} \cup R_{i+1} \cdots \cup R_N \right)$
denote the index set of coded symbols that are queried only by $\rx_i$. Also, let
$\Mc_i = R_i \cap \left(R_1 \cup \cdots \cup R_{i-1} \cup R_{i+1} \cdots \cup R_N \right)$
denote the index set of coded symbols that are queried by $\rx_i$ and at least one other receiver. Note that $\Sc_i \cap \Mc_i = \phi$ and $\Sc_i \cup \Mc_i = R_i$ for each $i \in [N]$.

The following result shows that certain entries of the encoder matrix can be assumed to be equal to zero without affecting the locality or rate of the linear index code.

\begin{tcolorbox}[colback=white,boxsep=3pt,left=2pt,right=2pt,top=2pt,bottom=2pt]
\begin{theorem} \label{thm:zeroes}
Let $\pmb{L} \in \Fb_q^{MN \times \ell}$ be a valid encoding matrix corresponding to the queries $R_1,\dots,R_N$. Then there exists a valid encoding matrix $\pmb{L}' \in \Fb_q^{MN \times \ell}$ for the queries $R_1,\dots,R_N$ such that for each $i \in [N]$
\begin{equation*}
 \pmb{L}'_k \lhd \Db_i \text{ for all } k \in \Sc_i.
\end{equation*} 
\end{theorem}
\end{tcolorbox}
\begin{proof}
See Appendix~\ref{app:proof:thm:zeroes}.
\end{proof}

The new index code guaranteed by Theorem~\ref{thm:zeroes} employs the same code length and the same set of queries as the given index code. Hence, the broadcast rate $\beta$, overall locality $r$ and the average locality $\ravg$ of the new code are identical to those of the given index code. Additionally, the new code guarantees that for any $i \in [N]$ any codeword symbol $c_k$, $k \in \Sc_i$, queried only by $\rx_i$, can be expressed as a linear combination of the demands of $\rx_i$. 
Since any valid encoder matrix can be modified to satisfy this property using Theorem~\ref{thm:zeroes}, in the sequel, without loss of generality, we will only consider encoder matrices $\pmb{L}$ that satisfy
\begin{equation} \label{eq:zeroes}
 \pmb{L}_k \lhd \Db_i \text{ for all } k \in \Sc_i \text{ and } i \in [N].
\end{equation} 
Further, without loss of generality, we will assume that for each $i \in [N]$, the vectors $\pmb{L}_k$, $k \in R_i$, are linearly independent. If this is not the case, then at least one of the codeword symbols $c_k$ queried by $\rx_i$ is a linear combination of the other queried symbols $\pmb{c}_{R_i \setminus \{k\}}$. 
Reducing the index set of the queries of $\rx_i$ from $R_i$ to $R_i \setminus \{k\}$ does not affect decodability at $\rx_i$ since $c_k$ can be reconstructed from $\pmb{c}_{R_i \setminus \{k\}}$. 
Note that this reduction in the queries does not increase the value of either the overall locality $r$ or the average locality $\ravg$ of the index code.
This process can be repeated till the columns of $\pmb{L}$ corresponding to the queries of each of the receivers are linearly independent. 
Finally, any codeword symbol that is not queried by any of the receivers can be removed from the transmission since this symbol will not be used for decoding.

Let us denote the index set of codeword symbols that are queried exactly once by
\begin{equation}
\Sc = \Sc_1 \cup \cdots \cup \Sc_N.
\end{equation} 
Note that $\Sc_i \cap \Sc_j = \phi$ for any $i \neq j$. Hence, $|\Sc|=\sum_{i=1}^{N}|\Sc_i|$.
Let 
\begin{equation}
\Mc = \Mc_1 \cup \cdots \cup \Mc_N 
\end{equation} 
denote the index set corresponding to the codeword symbols that have been queried by more than one receiver. Observe that $\Sc \cap \Mc = \phi$ and $\Sc \cup \Mc = [\ell]$.

\begin{lemma} \label{lem:single_queries}
For any valid index code for message length $M$, number of receivers $N$, rate $\beta$ and average locality $\ravg$, 
\begin{equation*}
 |\Sc| \geq M(2\beta - N\ravg),
\end{equation*} 
where $|\Sc|$ is the number of codeword symbols that have been queried exactly once.
\end{lemma}
\begin{proof}
We will count the total number of queries made by all the receivers in two different ways and relate these expressions to arrive at the statement of this lemma. 
The number of queries made by $\rx_i$ is $|R_i|$. Hence, the total number of queries made by all the receivers is $\sum_{i \in [N]}|R_i|$. From~\eqref{eq:ravg} this is equal to $MN\ravg$. The number of times a codeword symbol $c_k$ is queried is equal to $1$ if $k \in \Sc$, and is at least $2$ if $k \in \Mc$. Thus the total number of queries satisfies
\begin{align*}
 MN\ravg &\geq \sum_{k \in \Sc} 1 \, + \, \sum_{k \in \Mc} 2= |\Sc| + 2|\Mc| \\ 
&= |\Sc| + 2(\ell - |\Sc|) = 2\ell - |\Sc|,
\end{align*} 
where we have used the fact $|\Sc| + |\Mc|=\ell$. Substituting $\ell = M\beta$ in the above inequality, we arrive at $MN\ravg \geq 2M\beta - |\Sc|$ thereby proving the lemma.
\end{proof}

\subsection{Optimal Linear Codes for Directed Cycles} \label{sec:sub:directed_cycles}

We will now consider the index coding problem where $G$ is a directed $N$-cycle, i.e., for $i=1,\dots,N-1$, $K_i = \{i+1\}$ and $K_N=\{1\}$. 
In this subsection we will provide the proof of the following result.

\begin{tcolorbox}[colback=white,boxsep=3pt,left=2pt,right=2pt,top=2pt,bottom=2pt]
\begin{theorem} \label{thm:directed_cycles}
For any $N \geq 3$, let $G$ be a directed cycle of length $N$. For any finite field $\Fb_q$, the optimal trade-off between rate and locality among (vector) linear index codes for $G$ over $\Fb_q$ is 
\begin{equation} \label{eq:N_cycle_tradeoff}
\beta^*_{G,q}(r) = \max\left\{\frac{N(N-1-r)}{N-2},N-1\right\}, ~~~r \geq 1.
\end{equation} 
\end{theorem}
\end{tcolorbox}

Note that locality \mbox{$r \geq 1$} for any valid index coding scheme, and hence, $\beta^*_{G,q}(r)$ is defined for \mbox{$r \geq 1$} only.
\begin{figure}[!t]
\centering
\includegraphics[width=2in]{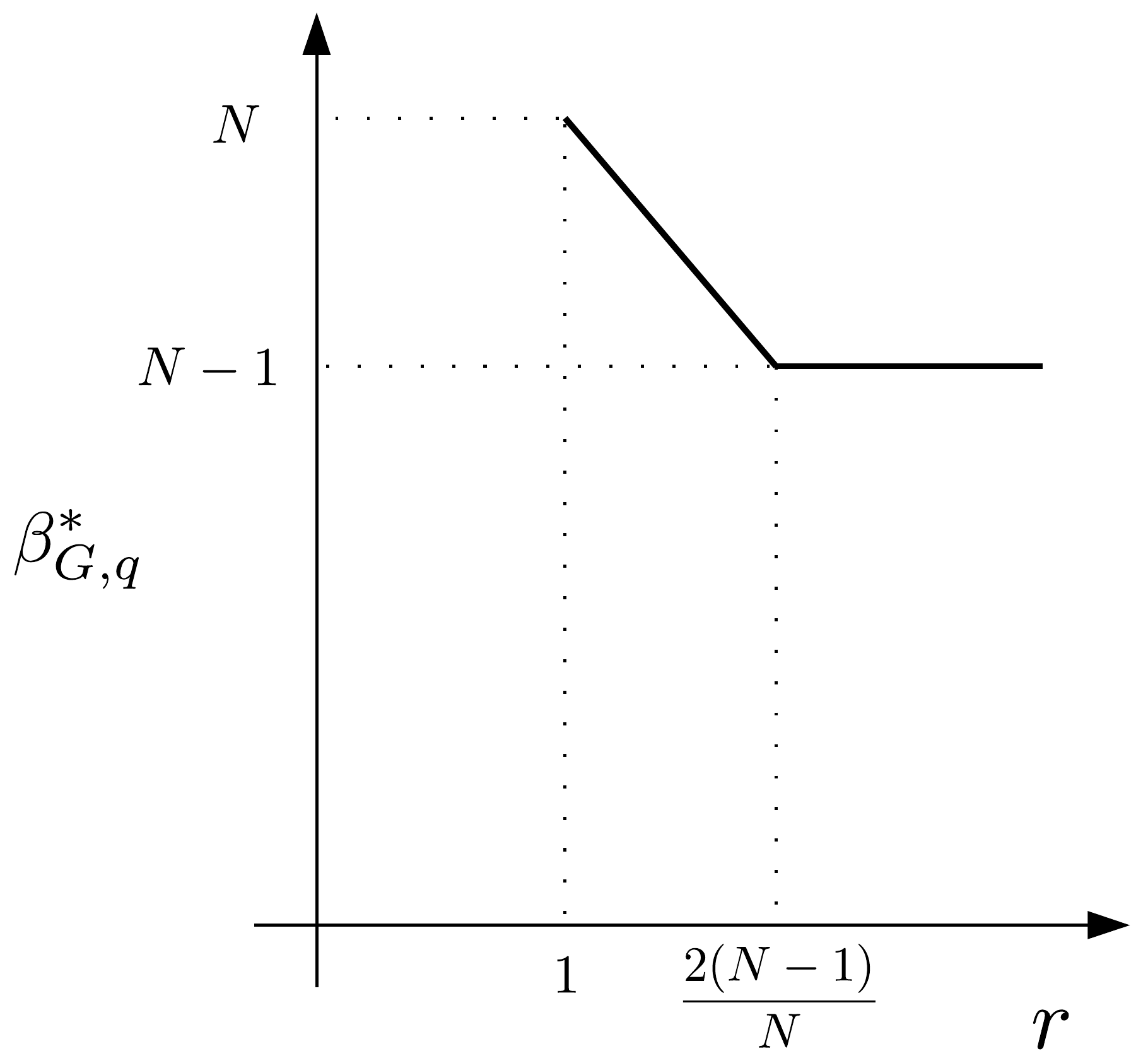}
\caption{The locality-rate trade-off of linear index codes for the directed $N$-cycle index coding problem.}
\label{fig:N_cycle_tradeoff}
\end{figure} 
The trade-off between rate and locality is shown in Fig.~\ref{fig:N_cycle_tradeoff}. 
Sections~\ref{sec:converse_directed_cycle} and~\ref{sec:achievability_directed_cycle} provide the proofs for the converse and achievability, respectively, of this rate-locality trade-off.
It is known that for the directed $N$-cycle $G$, $N-1$ is the minimum possible broadcast rate\footnote{From~\cite[Theorem~3]{YBJK_IEEE_IT_11}, we know that the broadcast rate is lower bounded by the size of the maximum acyclic induced subgraph of $G$, which is equal to $N-1$ in this case. Further, Example~1 of~\cite{YBJK_IEEE_IT_11} provides a scalar linear code that achieves rate $N-1$ for directed $N$-cycles.}.
Our result shows that the smallest locality at which this rate is achievable using linear codes is $r=2(N-1)/N$. 
This locality is achievable if the message length $M$ of the vector linear code is chosen carefully.
In Section~\ref{sec:locality_and_message_length}, we provide a detailed analysis of the effect of the message length $M$ on the locality $r$ when the broadcast rate is $N-1$.

\subsubsection{Converse} \label{sec:converse_directed_cycle}

In this subsection we will show that $\beta^*_{G,q}(r)$ is lower bounded by both $N-1$ and $N(N-1-r)/(N-2)$. It is clear that $\beta^*_{G,q}(r) \geq N-1$, since even without any locality constraints the smallest possible broadcast rate for the directed $N$-cycle is $N-1$. To complete the converse, we only need to show that $\beta^*_{G,q}(r) \geq N(N-1-r)/(N-2)$.

Suppose that the encoding matrix $\pmb{L}$ is valid with respect to a set of queries $R_1,\dots,R_N$. From Lemma~\ref{lem:single_queries}, the number of codeword symbols queried exactly once
$|\Sc| = \sum_{i \in [N]}|\Sc_i| \geq M(2\beta-N\ravg)$.
Hence there exists an $i \in [N]$ such that $|\Sc_i| \geq M(2\beta -N\ravg)/N$. Without loss of generality, let us assume that 
\begin{equation} \label{eq:cycle_converse:2}
|\Sc_N| \geq M(2\beta - N\ravg)/N. 
\end{equation} 
We now relate this lower bound on $|\Sc_N|$ to the rank of $\pmb{L}$ to complete the converse.

For $i=1,\dots,N-1$, we have $\Kb_i=\{iM+1,iM+2,\dots,(i+1)M\}$ and $\Db_i = \{(i-1)M+1,\dots,iM\}$. From Theorem~\ref{thm:design_criterion}, for each $j \in \Db_i$ there exists a $\pmb{u}_j \lhd \Kb_i$ such that $\pmb{e}_j + \pmb{u}_j \in \colsp(\pmb{L})$, where $\colsp(\pmb{L})$ denotes the column span of the matrix $\pmb{L}$. 
The non-zero entries of the $MN \times M$ matrix $[\pmb{u}_j,~ j \in \Db_i]$ are restricted to the $M$ rows indexed by $\Kb_i$. Let $\pmb{C}_i$ be the $M \times M$ submatrix of $[\pmb{u}_j,~ j \in \Db_i]$ corresponding to the rows $\Kb_i$.
Considering the first $N-1$ receivers $i=1,\dots,N-1$ and each of their demands $j \in \Db_i$, we obtain $M(N-1)$ vectors $\pmb{e}_j + \pmb{u}_j$, all which lie in $\colsp(\pmb{L})$. Arranging these column vectors into a matrix of size $MN \times M(N-1)$ we arrive at
\begin{equation} \label{eq:cycle_converse:1}
          \begin{bmatrix} 
          \pmb{I} & \pmb{0} & \cdots & \pmb{0} & \pmb{0} \\
          \pmb{C}_1 & \pmb{I} & \cdots & \pmb{0} & \pmb{0} \\
          \pmb{0} & \pmb{C}_2 & \cdots & \pmb{0} & \pmb{0} \\
          \vdots & \vdots &            & \vdots & \vdots  \\
          \pmb{0} & \pmb{0} & \cdots & \pmb{I} & \pmb{0} \\
          \pmb{0} & \pmb{0} & \cdots & \pmb{C}_{N-2} & \pmb{I} \\
          \pmb{0} & \pmb{0} & \cdots & \pmb{0} & \pmb{C}_{N-1}
          \end{bmatrix},
\end{equation} 
where each of the submatrices is of size $M \times M$. Note that the columns of this matrix are linearly independent.

Now considering $\rx_N$, we note that $\Db_N=\{(N-1)M+1,\dots,MN\}$. The coded symbols $\pmb{x}^\tr\pmb{L}_k$, $k \in \Sc_N$, are queried only by $\rx_N$. Using Theorem~\ref{thm:zeroes} and~\eqref{eq:zeroes}, we assume without loss of generality that $\pmb{L}_k \lhd \Db_N$ for all $k \in \Sc_N$, i.e., $\supp(\pmb{L}_k) \subseteq \Db_N$. Since the set of vectors $\{\pmb{L}_k \, | \, k \in R_N\}$ is linearly independent and $\Sc_N \subseteq R_N$, we observe that the vectors $\pmb{L}_k$, $k \in \Sc_N$, are linearly independent as well. Note that each of these vectors is a column of $\pmb{L}$ and hence lies in $\colsp(\pmb{L})$. Appending these $|\Sc_N|$ vectors as columns to the matrix in~\eqref{eq:cycle_converse:1}, we arrive at the block matrix
\begin{equation*}
 \pmb{A} = \begin{bmatrix} 
          \pmb{I} & \pmb{0} & \cdots & \pmb{0} & \pmb{0}  & \pmb{0} \\
          \pmb{C}_1 & \pmb{I} & \cdots & \pmb{0} & \pmb{0} & \pmb{0}  \\
          \pmb{0} & \pmb{C}_2 & \cdots & \pmb{0} & \pmb{0} & \pmb{0}  \\
          \vdots & \vdots &            & \vdots & \vdots & \pmb{0}   \\
          \pmb{0} & \pmb{0} & \cdots & \pmb{I} & \pmb{0} & \pmb{0}  \\
          \pmb{0} & \pmb{0} & \cdots & \pmb{C}_{N-2} & \pmb{I} & \pmb{0}  \\
          \pmb{0} & \pmb{0} & \cdots & \pmb{0} & \pmb{C}_{N-1} & \pmb{B}_N
          \end{bmatrix},
\end{equation*} 
where $\pmb{B}_N$ is an $M \times |\Sc_N|$ matrix with linearly independent columns. 
Note that each column of $\pmb{A}$ lies in $\colsp(\pmb{L})$, i.e., \mbox{$\colsp(\pmb{A}) \subseteq \colsp(\pmb{L})$}, and the columns of $\pmb{A}$ are linearly independent, i.e., $\rank(\pmb{A}) = M(N-1) + |\Sc_N|$.
Thus, 
\begin{align*}
\ell &\geq \rank(\pmb{L}) \geq \rank(\pmb{A}) ~~~~~~~~~~~~~ (\text{since } \colsp(\pmb{L}) \supseteq \colsp(\pmb{A})) \\
&= M(N-1) + |\Sc_N| \\
&\geq M(N-1) + M(2\beta - N\ravg)/N ~~~~~~~(\text{using~\eqref{eq:cycle_converse:2}})
\end{align*} 
Since broadcast rate \mbox{$\beta=\ell/M$}, the above inequality yields \mbox{$\beta \geq (N-1) + (2\beta - N\ravg)/N$}, which upon manipulation results in
\begin{equation} \label{eq:ravg_beta_bound_N_cycle}
\beta \geq {N(N-1-\ravg)}/{(N-2)}.
\end{equation} 
Since $\ravg \leq r$, we arrive at $\beta \geq {N(N-1-r)}/{(N-2)}$.

\subsubsection{Achievability} \label{sec:achievability_directed_cycle}

In this subsection we show that the trade-off in~\eqref{eq:N_cycle_tradeoff} is achievable using linear index codes. We will show that the points $(r,\beta)=(1,N)$ and $(r,\beta)=(2(N-1)/N,N-1)$ are achievable. Then any point on the line segment $\beta = N(N-1-r)/(N-2)$, $1 \leq r \leq 2(N-1)/N$ can be achieved using time sharing between these two schemes.
The achievability of the points $\beta=N-1$ and $r>2(N-1)/N$ will follow immediately since the rate $N-1$ is already achievable with $r=2(N-1)/N$.

\subsubsection*{Achieving $r=1$, $\beta=N$}

The point $(r,\beta)=(1,N)$ can be achieved trivially using the uncoded scheme, i.e., the transmitted codeword equals the message vector $\pmb{c}=\pmb{x} \in \Fb_q^{MN}$. Each receiver $\rx_i$ queries $\pmb{c}_{\Db_i}=\pmb{x}_{\Db_i}$ to meet its demand. 

\subsubsection*{Achieving $r=2(N-1)/N$, $\beta=N-1$}

Example~\ref{ex:simple_code_cycle} provides a family of $N$ scalar linear codes for $G$, one for each choice of $i\in [N]$, with rate $\beta=N-1$. The $i^\tth$ code provides localities $r_i=r_{i+1}=1$ and $r_j=2$ for all $j \neq i,i+1$, where we interpret $i+1$ as $1$ if $i=N$. Using $\pmb{r}=(r_1,r_2,\dots,r_N)$ to represent the tuple of receiver localities, we observe that rate $N-1$ can be achieved using scalar linear codes for the following values of locality vector $\pmb{r}$
\begin{align}
\pmb{r}_1&=(1,1,2,2,\dots,2),~ \pmb{r}_2=(2,1,1,2,\dots,2),\dots, \nonumber \\
\pmb{r}_{N-1}&=(2,2,\dots,2,1,1),~ \pmb{r}_N=(1,2,\dots,2,1). \label{eq:r_tuples}
\end{align}  

If $N$ is an odd integer, we time share the $N$ scalar linear codes corresponding to $\pmb{r}_1,\pmb{r}_2,\dots,\pmb{r}_N$. Observe that the overall scheme is a vector linear code for message length $M=N$, rate $N-1$ and locality $r=\ravg=2(N-1)/N$.
If $N$ is an even integer, we time share $N/2$ scalar linear codes corresponding to $\pmb{r}_1,\pmb{r}_3,\dots,\pmb{r}_{N-1}$, that yields a vector linear code with $M=N/2$, rate $N-1$ and $r=\ravg=2(N-1)/N$.

\subsubsection{Dependence of locality on message length} \label{sec:locality_and_message_length}

From the achievability scheme in Section~\ref{sec:achievability_directed_cycle} we observe that for $\beta=N-1$, the locality $r=\ravg=2(N-1)/N$ can be achieved using message length $M=N$ if $N$ is odd, and $M=N/2$ if $N$ is even.
We will now show that the message length used by the proposed scheme is the minimum required to attain locality $2(N-1)/N$ with rate $N-1$.

\begin{lemma} \label{lem:minimum_M_for_locality}
Let $N \geq 3$. The message length $M$ of any index code with locality equal to $2(N-1)/N$ for the directed $N$-cycle satisfies $M \geq N$ if $N$ is odd, and $M \geq N/2$ if $N$ is even.
\end{lemma}
\begin{proof}
Consider any valid coding scheme with locality $r=2(N-1)/N$. There exists an $i \in [N]$ such that
\begin{equation*}
\frac{2(N-1)}{N} = r = r_i = \frac{|R_i|}{M},
\end{equation*} 
that is 
\begin{equation*}
 M = \frac{|R_i|\,N}{2(N-1)}.
\end{equation*} 
If $N$ is odd, $N$ and $2(N-1)$ have no common factors, and since $M$ is an integer, we deduce that $M$ must be a multiple of $N$, i.e., $M \geq N$. If $N$ is even, using the fact $N/2$ and $N-1$ have no common factors we arrive at $M \geq N/2$.
\end{proof}

From Lemma~\ref{lem:minimum_M_for_locality}, it is clear that the minimum $M$ required to attain $r=2(N-1)/N$ at rate $N-1$ is $M=N$ if $N$ is odd and $M=N/2$ if $N$ is even. We will now derive the optimal locality when the message length is smaller than this quantity, i.e., $M<N$. We do so by analysing the two cases, $M < N/2$ and $N/2 \leq M < N$.

\subsubsection*{Locality when $M < N/2$}

From~\eqref{eq:ravg_beta_bound_N_cycle} we deduce that for any vector linear scheme of rate $\beta=N-1$, we have 
\begin{equation*}
 \ravg \geq 2(N-1)/N.
\end{equation*} 
Thus, $\sum_{i=1}^{N}|R_i| = MN\ravg \geq 2M(N-1)$. It follows that there exists an $i \in [N]$ such that 
\begin{equation} \label{eq:size_of_Ri_M_and_N}
|R_i| \geq \frac{2M(N-1)}{N} = 2M - \frac{M}{N/2}.
\end{equation} 
If $M<N/2$, considering the fact that $|R_i|$ is an integer, we deduce that $|R_i| \geq 2M$. Hence, $r_i=|R_i|/M \geq 2$, and thus, $r \geq 2$. This lower bound on $r$ can be achieved by simply using the scalar linear code of Example~\ref{ex:simple_code_cycle} $M$ times, leading to a vector linear code for message length $M$, rate $N-1$ and $r=2$. Note that this code still achieves the optimal value of average locality $\ravg=2(N-1)/N$.

\subsubsection*{Locality when $N/2 \leq M < N$}

If $N$ is even, the message length $M=N/2$ is sufficient to attain $r=2(N-1)/N$. Thus it is enough to consider larger values of $M$, i.e., \mbox{$N/2 \leq M < N$} only for $N$ odd.
From~\eqref{eq:size_of_Ri_M_and_N} and using the fact that $|R_i|$ is an integer, we arrive at $|R_i| \geq 2M-1$. Thus, 
\begin{equation*}
r \geq r_i \geq 2 - \frac{1}{M}.
\end{equation*} 
Assuming $N$ is odd, this lower bound on $r$ is achieved by time sharing the $M$ scalar linear codes from Section~\ref{sec:achievability_directed_cycle} corresponding to the tuples of localities 
\begin{equation*}
\pmb{r}_1,\pmb{r}_3,\dots,\pmb{r}_{N-3},\pmb{r}_N,\pmb{r}_2,\pmb{r}_4,\dots,\pmb{r}_{2M-(N+1)},
\end{equation*} 
see~\eqref{eq:r_tuples}. It is straightforward to show that this scheme has rate $N-1$, $r=2\,-\,1/M$ and $\ravg=2(N-1)/N$. 
Note that in the interval $N/2 \leq M < N$, the value of the optimal locality increases with $M$. Hence, the choice $M=(N+1)/2$ yields the smallest locality in this interval. 
\begin{example}
Consider the problem of designing a vector linear index code for the directed $5$-cycle, i.e., $N=5$, with message length $M=3$.
Note that $N$ is odd and $N/2 \leq M$.
From Example~\ref{ex:simple_code_cycle} and~\eqref{eq:r_tuples}, we know that there exist three scalar linear encoders, each of rate $4$, for this index coding problem for which the tuple of receiver localities are
\begin{equation*}
\pmb{r}_1 = (1,1,2,2,2),\,\pmb{r}_3 = (2,2,1,1,2)~\text{and}~\pmb{r}_5=(1,2,2,2,1).
\end{equation*}
To arrive at a vector linear scheme of message length $M=3$, we consider three generations of scalar messages and encode them using the above three scalar linear index coding schemes, respectively. The overall rate of this time-sharing scheme is $\beta=4$. The total number of codeword symbols observed by each receiver for this vector linear code is $|R_1|=4$, $|R_2|=5$, $|R_3|=5$, $|R_4|=5$ and $|R_5|=5$. Thus, this scheme has $r=5/3$ and $\ravg=8/5$, which are equal to $2-1/M$ and $2(N-1)/N$, respectively.
\end{example}

\section{Locality-Broadcast Rate Trade-Off \\ of Directed $3$-Cycle} \label{sec:3cycle}

Let $G$ be the directed $3$-cycle i.e., $N=3$ and $K_1=\{2\}$, $K_2=\{3\}$ and $K_3=\{1\}$.
In this section we identify the optimal locality-rate trade-off function $\beta_G^*(r)$. 
The converse presented in this section applies to any valid index code (including non-linear codes), while the achievability proof uses the vector linear scheme of Section~\ref{sec:sub:directed_cycles}.
Note that for this index coding problem the trade-off curves $\beta_G^*(r)$ and $\beta_{G,q}^*(r)$ are identical. 
%
The main result of this section is
\begin{tcolorbox}[colback=white,boxsep=3pt,left=2pt,right=2pt,top=2pt,bottom=2pt]
\begin{theorem} \label{thm:3cycle}
For the directed $3$-cycle $G$, the optimal locality-broadcast rate trade-off function (considering non-linear codes also) is 
\begin{equation*} 
 \beta_G^*(r) = \max\{6-3r,2\} \text{ for all } r \geq 1.
\end{equation*} 
\end{theorem}
\end{tcolorbox}

Using $N=3$ in~\eqref{eq:N_cycle_tradeoff}, we have $\beta_{G,q}^*(r)=\max\{6-3r,2\}$ for $r \geq 1$. Achievability follows from observing that $\beta_G^*(r) \leq \beta_{G,q}^*(r)$, where the right hand side of the inequality considers only linear index codes.

In the rest of this section we prove the converse for Theorem~\ref{thm:3cycle}.
Towards proving the converse, we first present a result that exploits the symmetry in the side information graph.
Let $G$ be any single unicast index coding problem involving $N$ messages such that the cyclic permutation $\sigma$ on $[N]$ that maps $i \in [N]$ to $\sigma(i) = (i \mod N) + 1$ is an automorphism of $G$. 
Since we are considering non-linear codes too, an index code for $G$ is represented by the encoding function $\Ef$ and the decoding functions $\Df_1,\dots,\Df_N$ defined over a finite alphabet $\Ac$. 
The encoder is a function $\Ef:\Ac^{MN} \to \Ac^\ell$ that maps the message vectors $\xb_1,\dots,\xb_N \in \Ac^M$ to the codeword $\cb \in \Ac^\ell$. The decoder used at $\rx_i$ is a function $\Df_i:\Ac^{|R_i|} \times \Ac^{M|K_i|} \to \Ac^M$ which maps the vectors $\cb_{R_i}$ and $\xb_j$, $j \in K_i$, to $\xb_i$.
We represent an index code by the tuple $(\Ef,\Df_1,\dots,\Df_N)$ of encoding and decoding functions.

\begin{lemma} \label{lem:cyclic_symmetry}
Let the cyclic permutation $\sigma$ be an automorphism of $G$,
and $(\Ef,\Df_1,\dots,\Df_N)$ be a valid index code for $G$ with rate $\beta$ and locality $r$.
Then there exists a valid index code $(\Ef',\Df_1',\dots,\Df_N')$ for $G$ with rate $\beta$ and locality at the most $r$ such that the index sets of codeword symbols observed by the $N$ receivers $R_1',\dots,R_N'$ for this code satisfy:
\begin{enumerate}
\item[\emph{(i)}] $|R_1'|=|R_2'|\cdots=|R_N'|$; and 
\item[\emph{(ii)}] $|R_1' \cap R_2'| = |R_2' \cap R_3'| = \cdots = |R_N' \cap R_1'|$.
\end{enumerate} 
\end{lemma} 
\begin{proof}
See Appendix~\ref{app:lem:cyclic_symmetry}.
\end{proof}

Lemma~\ref{lem:cyclic_symmetry} shows that, without affecting the broadcast rate and the locality, we can assume that the receiver queries in any valid index code for $G$ satisfy the two symmetry properties listed above.
In the rest of this section we will assume that $G$ is a directed $3$-cycle, and $(\Ef,\Df_1,\Df_2,\Df_3)$ is a valid index code for $G$ such that $|R_1|=|R_2|=|R_3|$ and $|R_1 \cap R_2| = |R_2 \cap R_3|=|R_3 \cap R_1|$.

For the sake of brevity, we abuse the notation mildly by using $i+1$ to denote the receiver index $(i\!\! \mod 3) + 1$, 
and similarly we use $i+2$ to denote $\big((i+1) \!\! \mod 3\big) + 1$. 
With this notation, for $i=1,2,3$, the side information index set of the $i^{\text{th}}$ receiver is $K_i=\{i+1\}$.
Assume, as usual, that the messages $\xb_1,\xb_2,\xb_3$ are independently and uniformly distributed in $\Ac^M$.

Now considering the decoding operation at $\rx_i$, we have \mbox{$I(\xb_i;\cb_{R_i}|\xb_{i+1})=H(\xb_i)=M$}. Expanding this term as a difference of conditional entropies, we have
$H(\cb_{R_i}|\xb_{i+1}) - H(\cb_{R_i}|\xb_i,\xb_{i+1}) = M$. Using  this with the upper bound $H(\cb_{R_i}|\xb_{i+1}) \leq H(\cb_{R_i}) \leq |R_i|$, we arrive at
\begin{equation*}
H(\cb_{R_i}|\xb_i,\xb_{i+1}) \leq |R_i| - M.
\end{equation*} 
Using the above inequality and the fact that $\cb_{R_i}$ is a deterministic function of all three messages $\xb_i,\xb_{i+1},\xb_{i+2}$, we have $I(\cb_{R_i};\xb_{i+2}|\xb_i,\xb_{i+1})=H(\cb_{R_i}|\xb_i,\xb_{i+1}) \leq |R_i|-M$. Hence,
\begin{align*}
H(\xb_{i+2}|\xb_i,\xb_{i+1}) - H(\xb_{i+2}|\cb_{R_i},\xb_i,\xb_{i+1}) \leq |R_i| - M.
\end{align*} 
Since $H(\xb_{i+2}|\xb_i,\xb_{i+1})=M$, we obtain the lower bound
\begin{equation} \label{eq:3cycle_conv_1}
H(\xb_{i+2}|\cb_{R_i},\xb_i,\xb_{i+1}) \geq 2M-|R_i|. 
\end{equation} 

Our objective now is to use the above inequality to obtain an upper bound on $|R_i \cap R_{i+2}|$, which can then be translated into a lower bound on $\ell$, and hence, a lower bound on $\beta$. 
To do so, observe that $\cb_{R_i}$ is composed of the following two sub-vectors $\cb_{R_i \cap R_{i+2}}$ and $\cb_{R_i \setminus R_{i+2}}$. 
Using~\eqref{eq:3cycle_conv_1}, we obtain
\begin{align*}
H(\xb_{i+2}|\cb_{R_i \cap R_{i+2}},\xb_i) \! \geq H(\xb_{i+2}|\cb_{R_i},\xb_i,\xb_{i+1}) \! \geq 2M-|R_i|.
\end{align*} 
Using this inequality, and the relation $H(\xb_{i+2}|\cb_{R_{i+2}},\xb_i)=0$ (to satisfy the demands of $\rx_{i+2}$), we obtain the following
\begin{align*}
|R_{i+2} \setminus R_i| &\geq H(\cb_{R_{i+2} \setminus R_i}) 
\geq I(\xb_{i+2};\cb_{R_{i+2} \setminus R_i}|\cb_{R_i \cap R_{i+2}},\xb_i) \\
&= H(\xb_{i+2}|\cb_{R_i \cap R_{i+2}},\xb_i) \\ 
&~~~~~~~~~~~~~~~~ - H(\xb_{i+2}|\cb_{R_{i+2}\setminus R_i},\cb_{R_i \cap R_{i+2}},\xb_i) \\
&= H(\xb_{i+2}|\cb_{R_i \cap R_{i+2}},\xb_i) - H(\xb_{i+2}|\cb_{R_{i+2}},\xb_i) \\
&\geq 2M - |R_i|.
\end{align*} 
Since $|R_1|=|R_2|=|R_3|$, we now have
\begin{align*}
|R_i \cap R_{i+2}| &= |R_{i+2}| - |R_{i+2} \setminus R_i| \\
& \leq |R_i| - (2M - |R_i|) \\
& = 2\left(|R_i|-M \right).
\end{align*} 
Finally, since $|R_i \cap R_{i+2}|$ is independent of $i$, we have
\begin{align*}
\ell = |R_1 \cup R_2 \cup R_3| 
&\geq \sum_{j=1}^{3} |R_j| - \sum_{j=1}^{3} |R_j \cap R_{j+2}| \\ 
&= 3|R_i| - 3|R_i \cap R_{i+2}| \\
&\geq 3|R_i| - 3 \times 2(|R_i|-M) \\
&= 6M - 3|R_i|.
\end{align*} 
Dividing both sides by the message length $M$, and remembering that all the receivers have the same locality $r=r_1=r_2=r_3$, we have $\beta=\ell/M \geq 6-3r$. Thus we have 
\begin{align*}
\beta_G^*(r) \geq 6-3r \text{ for all } r \geq 1.
\end{align*} 
Further, the minimum possible broadcast rate $\betaopt(G)=2$, and hence, $\beta_G^*(r) \geq 2$ for all $r \geq 1$. Combining this with the above inequality we have arrived at the converse
$\beta_G^*(r) \geq \max\{6-3r,2\}$ for all $r \geq 1$.

\begin{remark} \label{rem:3cycle}
The main idea behind the derivation of $\beta_G^*(r)$ for the directed $3$-cycle in this section is to lower bound the codelength using the identity
$\ell = |R_1 \cup R_2 \cup R_3| \geq \sum_{i=1}^{3} |R_i| \, - |R_1 \cap R_2| - |R_2 \cap R_3| - |R_3 \cap R_1|$. 
The symmetry of the side information graph implies, through Lemma~\ref{lem:cyclic_symmetry}, that it is enough to find or bound the values of $|R_1|$ and $|R_1 \cap R_2|$.
The other terms in this inequality are equal to either $|R_1|$ or $|R_1 \cap R_2|$.
%
In contrast, for directed $N$-cycles with $N>3$, we have
\begin{equation} \label{rep:eq:N_cycle}
 \ell = \left|\cup_{i=1}^{N} R_i\right|  \geq \sum_{i=1}^{N} |R_i| - \sum_{i \neq j} |R_i \cap R_j|. 
\end{equation}
The symmetry of the side information graph and the results in Lemma~\ref{lem:cyclic_symmetry} seem to be only partially helpful here. 
Although it is true that $|R_1|=\cdots=|R_N|$ (Lemma~\ref{lem:cyclic_symmetry}, property~\emph{(i)}), there is no guarantee that the terms $|R_i \cap R_j|$, $i \neq j$, are all equal.
It is possible that $|R_1 \cap R_2| \neq |R_1 \cap R_3|$ when $N > 3$.
In general,~\eqref{rep:eq:N_cycle} contains a large number of terms that need to be individually analyzed and bounded. 
\end{remark}

\section{Scalar Linear Codes for Large Minrank} \label{sec:minrank_N_minus_1}

In this section we analyze scalar linear codes that are locally decodable. If the minrank of the side-information graph $G$ is $N$, then $G$ is a directed acyclic graph and uncoded transmission is an optimal scalar linear index code. For this scheme, the locality of each receiver $r_i=1$. Thus, we observe that when $\minrk(G)=N$, we can achieve the optimal broadcast rate $\beta=N$ and the minimum possible locality $r = \ravg = 1$ simultaneously.
Thus, the problem of designing locally decodable index codes is interesting only when $\minrk(G) \leq N-1$.

In Section~\ref{sec:sub:scalar_structure} we identify some structural properties of scalar linear codes for arbitrary index coding problems $G$. 
These results are then specialized in Section~\ref{sec:sub:minrank_large} to the family of index coding problems for which the number of messages is one greater than minrank, and the optimal trade-off between broadcast rate and locality is derived for this family of problems.
Finally, in Section~\ref{sec:sub:receiver_localities_cycles} we consider directed $N$ cycles (minrank is $N-1$ for these problems) and characterize the set of all possible vectors of receiver localities $\pmb{r}=(r_1,\dots,r_N)$ that can be achieved through scalar linear coding when the broadcast rate is $N-1$.

\subsection{Preliminaries: Locally Decodable Scalar Linear Codes} \label{sec:sub:scalar_structure}

We will now derive a few properties of scalar linear index codes that will be useful in relating broadcast rate to locality. Note that, in this case message length $M=1$, demand set $\Db_i=\{i\}$ and $\Kb_i=K_i$ for all $i \in [N]$. 
The locality of each receiver $r_i=|R_i|/M = |R_i|$ is an integer, and so is the overall locality $r=\max_i r_i$. 

We say that a matrix $\pmb{A} \in \Fb_q^{N \times N}$ \emph{fits} $G=(\mathcal{V},\mathcal{E})$ if the diagonal elements of $\pmb{A}$ are all equal to $1$, and the $(j,i)^\tth$ entry of $\pmb{A}$ is zero if $j \notin K_i$, i.e., $\pmb{A}_{j,i}=0$ if $(i,j) \notin \mathcal{E}$. The \emph{minrank} of $G$ over $\Fb_q$ is the minimum among the ranks of all possible matrices $\pmb{A} \in \Fb_q^{N \times N}$ that fit $G$, and is denoted as $\minrank(G)$. It is known that the smallest possible scalar linear index coding rate is equal to $\minrank(G)$~\cite{YBJK_IEEE_IT_11,DSC_IT_12}. We also know from~\cite{YBJK_IEEE_IT_11,DSC_IT_12} that a matrix $\pmb{L}$ is a valid encoder matrix for $G$ if and only if for each receiver $i \in [N]$, there exists a vector $\pmb{u}_i \in \Fb_q^N$ such that $\pmb{u}_i \lhd K_i$ and $\pmb{u}_i + \pmb{e}_i \in \colsp(\pmb{L})$, where $\colsp$ denotes the column span of a matrix. 
If $\pmb{L}$ is a valid encoder matrix, stacking these vectors we obtain the $N \times N$ matrix 
\begin{equation*}
\pmb{A} = \left[ \pmb{u}_1 + \pmb{e}_1~~\pmb{u}_2+\pmb{e}_2~\cdots~\pmb{u}_N+\pmb{e}_N  \right].
\end{equation*} 
Notice that $\pmb{A}$ fits $G$ and $\colsp(\pmb{A}) \subseteq \colsp(\pmb{L})$. We will say that $\pmb{A}$ is a \emph{fitting matrix corresponding to the encoder matrix} $\pmb{L}$.

Suppose $\pmb{L} \in \Fb_q^{N \times \ell}$ is a valid scalar linear encoder and the queries of the $N$ receivers are $R_1,\dots,R_N \subseteq [\ell]$. 
From Theorem~\ref{thm:design_criterion}, for each $i \in [N]$, there exists a vector $\pmb{u}_i \lhd K_i$ such that $\pmb{u}_i + \pmb{e}_i \in \spp(\pmb{L}_k, k \in R_i)$.
Thus, there exist scalars $\alpha_{i,k}$, $k \in R_i$, such that $\pmb{u}_i+\pmb{e}_i = \sum_{k \in R_i} \alpha_{i,k}\pmb{L}_k$. 
The $i^\tth$ receiver decodes its demand by computing \mbox{$\sum_{k \in R_i}\alpha_{i,k}c_k - \pmb{x}^\tr \pmb{u}_i$}, which is equal to
\begin{align*}
\sum_{k \in R_i} \alpha_{i,k} \pmb{x}^\tr\pmb{L}_k - \pmb{x}^\tr\pmb{u}_i = \pmb{x}^\tr(\pmb{u}_i + \pmb{e}_i) - \pmb{x}^\tr\pmb{u}_i = x_i.
\end{align*} 
Notice that the receiver can compute $\pmb{x}^\tr\pmb{u}_i$ using its side information since $\supp(\pmb{u}_i) \subseteq K_i$.
We will assume that each scalar $\alpha_{i,k}$ is non-zero since if $\alpha_{i,k}=0$ the receiver does not need to query the coded symbol $c_k$.
Finally, notice that stacking the vectors $\pmb{u}_i + \pmb{e}_i$, $i \in [N]$, we obtain a fitting matrix $\pmb{A}$ corresponding to $\pmb{L}$.
  
For certain choices of $S \subseteq [N]$, we will now relate the sizes of $R_i$, $i \in S$, and their union $\cup_{i \in S}R_i$. Let $\nullsp(\pmb{A})$ denote the null space of $\pmb{A}$.

\begin{lemma} \label{lem:fitting_support}
Let $\pmb{A}$ be any fitting matrix corresponding to a valid scalar linear encoder $\pmb{L}$, $S \subseteq [N]$ be the support of a non-zero vector in $\nullsp(\pmb{A})$, and let the vectors $\pmb{L}_k$, $k \in \cup_{i \in S}R_i$, be linearly independent. Then
\begin{equation*}
\sum_{i \in S}|R_i| ~\geq~ 2\,\left| \bigcup_{i \in S} R_i \right|.
\end{equation*} 
\end{lemma}
\begin{proof}
Denote the columns of $\pmb{A}$ by $\pmb{A}_1,\dots,\pmb{A}_N$. Let $\pmb{z} \in \nullsp(\pmb{A}) \setminus \{\pmb{0}\}$ be such that $S = \supp(\pmb{z})$. Since $\pmb{Az}=\pmb{0}$, we have $\sum_{i \in S}z_i\pmb{A}_i=\pmb{0}$, where the components $z_i$, $i \in S$, of the vector $\pmb{z}$ are non-zero.
Notice that there exist non-zero scalars $\alpha_{i,k}$ such that $\pmb{A}_i = \sum_{k \in R_i} \alpha_{i,k} \pmb{L}_k$. Hence, we have
\begin{align*}
\pmb{0} = \sum_{i \in S}z_i\pmb{A}_i = \sum_{i \in S}\sum_{k \in R_i} z_i\alpha_{i,k} \pmb{L}_k.
\end{align*} 
All the scalars $z_i\alpha_{i,k}$ in the above linear combination are non-zero, and the set of vectors $\pmb{L}_k$, $k \in \cup_{i \in S}R_i$, appearing in this linear combination are linearly independent. Hence, this linear combination is zero only if each $\pmb{L}_k$, where $k \in \cup_{i \in S}R_i$, appears at least twice in the expansion $\sum_{i \in S}\sum_{k \in R_i} z_i\alpha_{i,k} \pmb{L}_k$, i.e., only if each $k \in \cup_{i \in S}R_i$ is contained in at least two distinct sets $R_i$ and $R_j$, $i \neq j$ and $i,j \in S$. Then a simple counting argument leads to the statement of this lemma.
\end{proof}


For any $S \subseteq [N]$, let $G_S$ denote the subgraph of $G$ induced by the vertices in $S$, i.e., the vertex set of $G$ is $S$ and edge set is $\{(i,j) \in \mathcal{E}~\vert~ i,j \in S\}$. The subgraph $G_S$ is the side information graph of the index coding problem obtained by restricting the index coding problem $G$ to the messages \mbox{$x_i$, $i \in S$}.
The following result can be used to manipulate the bound in Lemma~\ref{lem:fitting_support}.

\begin{lemma} \label{lem:lower_bound_on_union}
Let $\pmb{L}$ be the encoder matrix of a valid scalar linear index code for $G$ with receiver queries $R_1,\dots,R_N$. For any $S \subseteq [N]$, we have
\begin{equation*}
 \left| \bigcup_{i \in S} R_i \right| \geq \minrank(G_S).
\end{equation*} 
\end{lemma}
\begin{proof}
The submatrix $\pmb{L}_S$ of $\pmb{L}$ formed by the rows indexed by the set $S$ is a valid encoder matrix for the index coding problem $G_S$. 
Since the receivers $\rx_i$, $i \in S$, query only the coded symbols with indices \mbox{$k \in \cup_{i \in S}R_i$}, the submatrix of $\pmb{L}_S$ consisting of the columns with indices in $\cup_{i \in S}R_i$ is also a valid scalar linear encoder for $G_S$. Hence, the codelength $|\cup_{i \in S}R_i|$ of this index code is lower bounded by $\minrank(G_S)$.
\end{proof}

The next result follows immediately from Lemmas~\ref{lem:fitting_support} and~\ref{lem:lower_bound_on_union}.

\begin{tcolorbox}[colback=white,boxsep=3pt,left=2pt,right=2pt,top=2pt,bottom=2pt]
\begin{corollary} \label{cor:lower_bound_r}
If $\pmb{L}$ is an optimal scalar linear encoder for $G$, i.e., has codelength equal to $\minrank(G)$, $\pmb{A}$ is a fitting matrix corresponding to $\pmb{L}$ and $\pmb{z} \in \nullsp(\pmb{A}) \setminus \{\pmb{0}\}$, then
\begin{equation*}
\sum_{i \in S} r_i ~\geq~ 2 \,\minrank(G_S),
\end{equation*} 
where $S=\supp(\pmb{z})$.
\end{corollary}
\end{tcolorbox}
\begin{proof}
The matrix $\pmb{L}$ has linearly independent columns since the number of columns $\ell$ of $\pmb{L}$ satisfies
\begin{equation*}
\ell = \minrank(G) \leq \rank(\pmb{A}) \leq \rank(\pmb{L}) \leq \ell.
\end{equation*} 
The corollary holds since $r_i=|R_i|$ for scalar linear codes and the vectors $\pmb{L}_k$, $k \in \cup_{i \in S}R_i$ satisfy the conditions of Lemma~\ref{lem:fitting_support}.
\end{proof}

\subsection{Scalar Linear Coding when Minrank is $N-1$} \label{sec:sub:minrank_large}

In the rest of this section we will assume $\minrank(G)=N-1$ and characterize the optimal localities $r$ and $\ravg$ among scalar linear index codes for such index coding problems. 
The main result of this section is Theorem~\ref{thm:minrank_N-1}.

Recall that since $M=1$, the receiver localities $r_i$ and the rate $\beta$ are integers. 
Since the minimum scalar coding rate is equal to $\minrank$, we are interested in the operating points corresponding to $\beta=N$ and $\beta=N-1$. 
Note that the side information graph $G$ contains at least one directed cycle, since otherwise, $G$ is a directed acyclic graph and its minrank is equal to $N$~\cite{YBJK_IEEE_IT_11}, a contradiction. Let $N_c$ denote the length of the smallest directed cycle contained in $G$. 

If $N_c=2$, there exist $i,j \in [N]$ such that $(i,j),(j,i) \in \mathcal{E}$, i.e., \mbox{$i \in K_j$} and \mbox{$j \in K_i$}. The following scalar linear code attains the minimum possible locality $r=\ravg=1$ and the minimum possible rate \mbox{$\beta=N-1$} simultaneously. 
Transmit $x_i+x_j$ followed by transmitting the remaining $N-2$ information symbols uncoded. $\rx_i$ and $\rx_j$ can decode using $x_i+x_j$, and the remaining receivers query their demands directly from the codeword.

In the rest of this section we will assume that $N_c \geq 3$. Observe that rate $\beta=N$ can be achieved with smallest possible localities $r=\ravg=1$ using uncoded transmission. 
We will now consider the case $\beta=\minrank(G)=N-1$.
Let $\pmb{L}$ be any scalar encoder matrix with codelength $\ell=N-1$, and $\pmb{A}$ be a corresponding fitting matrix. Since $\ell=\minrank(G)$, we have
$\ell \leq \rank(\pmb{A}) \leq \rank(\pmb{L}) \leq \ell$,
and hence, $\rank(\pmb{A})=\rank(\pmb{L})=N-1=\ell$. Thus, the nullspace of $\pmb{A}$ contains a non-zero vector.

\begin{lemma} \label{lem:GS_contains_cycle}
If $\pmb{z} \in \nullsp(\pmb{A}) \setminus \{\pmb{0}\}$ and $S=\supp(\pmb{z})$, the subgraph $G_S$ of $G$ induced by the vertices $S$ contains at least one directed cycle.
\end{lemma}
\begin{proof}
Observe that the $|S| \times |S|$ submatrix $\pmb{A}'$ of $\pmb{A}$ composed of the rows and columns of $\pmb{A}$ with indices in $S$ fits $G_S$. Since $\pmb{Az}=\pmb{0}$, the columns of $\pmb{A}$ indexed by $S$ are linearly dependent. This implies that the columns of the submatrix $\pmb{A}'$ are linearly dependent as well, and hence, $\rank(\pmb{A}') \leq |S|-1$. It follows that
$\minrank(G_S) \leq \rank(\pmb{A}') \leq |S|-1$.
Clearly $G_S$ is not a directed acyclic graph since otherwise $\minrank(G_S)=|S|$.
\end{proof}

In order to use Corollary~\ref{cor:lower_bound_r}, we now derive a lower bound on $\minrank(G_S)$.

\begin{lemma} \label{lem:lower_bound_minrank_GS}
If $\minrank(G)=N-1$, then for any $S \subseteq [N]$, $\minrank(G_S) \geq |S|-1$.
\end{lemma}
\begin{proof}
Consider the following valid scalar linear code for $G$. Encode the information symbols \mbox{$x_i$, $i \in S$}, using the optimal scalar linear code for $G_S$, and use uncoded transmission for the remaining symbols. The length of this code is lower bounded by $\minrank(G)=N-1$, hence we obtain
$\minrank(G_S) + N - |S| \geq N - 1$.
\end{proof}

We now prove the main result of this section.

\begin{tcolorbox}[colback=white,boxsep=3pt,left=2pt,right=2pt,top=2pt,bottom=2pt]
\begin{theorem} \label{thm:minrank_N-1}
If $\minrank(G)=N-1$ and the smallest directed cycle in $G$ is of length $N_c \geq 3$, the optimal locality for scalar linear coding for $G$ with rate $N-1$ is
\begin{equation*}
 r = 2 \text{ and } \ravg = \frac{N+N_c-2}{N}.
\end{equation*} 
\end{theorem}
\end{tcolorbox}
\begin{proof}
\emph{Converse:} Let $\pmb{A}$ be a fitting matrix corresponding to any valid scalar linear code for $G$ with rate $N-1$. Let \mbox{$\pmb{z} \in \mathcal{N}(\pmb{A}) \setminus \{\pmb{0}\}$} and $S=\supp(\pmb{z})$.
From Corollary~\ref{cor:lower_bound_r} and Lemma~\ref{lem:lower_bound_minrank_GS}, $\sum_{i \in S}r_i \geq 2\,\minrank(G_S) \geq 2(|S|-1)$. Using the trivial bound $r_i \geq 1$ for $i \notin S$, we have
\begin{align*}
\sum_{i \in [N]} r_i &= \sum_{i \in S} r_i + \sum_{i \notin S}r_i \\
&\geq 2(|S|-1) + N-|S| = N + |S| - 2. 
\end{align*} 
From Lemma~\ref{lem:GS_contains_cycle}, we know that $G_S$ contains a cycle, and hence, the number of vertices $|S|$ in $G_S$ is at least $N_c$. Thus,
\begin{equation*}
\ravg = \frac{\sum_{i \in [N]}r_i}{N} \geq \frac{N +|S|-2}{N} \geq \frac{N+N_c-2}{N}.
\end{equation*} 
Since $N_c \geq 3$, we have $r \geq \ravg \geq (N+1)/N$, and since $r$ is an integer we conclude that $r \geq 2$.

\emph{Achievability:} Let $C \subseteq [N]$ be the set of vertices that form the smallest directed cycle in $G$. Note that the subgraph $G_C$ is a directed cycle of length $|C|=N_c$. We encode the symbols $x_i$, $i \in C$, using the scalar linear code given in Example~\ref{ex:simple_code_cycle} and send the remaining $N-N_c$ symbols uncoded. This achieves the codelength $N-1$. From Example~\ref{ex:simple_code_cycle}, the sum locality within the cycle $\sum_{i \in C}r_i = 2(N_c-1)$ and the maximum locality within the cycle $\max_{i \in C}r_i=2$. The locality of the remaining receivers is $r_i=1$, $i \notin C$. This scheme achieves the optimal values of $r$ and $\ravg$ for rate $N-1$.
\end{proof}

\begin{remark}
Corollary~V.2 of~\cite{KSCF_arXiv_18} shows that, over the binary field $q=2$, if $\minrank(G)=N-1$, then a scalar linear coding rate of $N-1$ is achievable for any choice of locality $r \geq 2$. 
Our results show that $r=2$ is optimal (if $N_c \geq 3$) and also provide the optimal value of the average locality $\ravg$ for scalar linear codes over an arbitrary finite field $\Fb_q$.
\end{remark}

\subsection{Feasible Receiver Localities for Directed Cycles} \label{sec:sub:receiver_localities_cycles}

We observe from the results in Section~\ref{sec:sub:minrank_large} that the interaction between rate and locality for scalar linear codes for graphs with minrank equal to $N-1$ depends on the rate-locality trade-off of directed cycles.
Further, one of the elementary index coding schemes (for arbitrary single unicast problems), known as \emph{cycle covering}~\cite{NTZ_IT_13,CASL_ISIT_11}, partitions the side information graph into disjoint cycles and applies a scalar linear code for each cycle.
In view of the important role directed cycles play in index coding, we now determine the possible values of the tuple of receiver localities $\p{r} = (r_1,\dots,r_N)$ for scalar linear coding when $G$ is a directed $N$-cycle with the desired broadcast rate being $N-1$.

\begin{definition}
A vector $\p{r}=(r_1,\dots,r_N)$ of receiver localities is \emph{feasible} for $G$ at rate $\ell$ through scalar linear coding if there exists a scalar linear code for $G$ with codelength at the most $\ell$ with the localities of $\rx_1,\dots,\rx_N$ at the most $r_1,\dots,r_N$, respectively.
\end{definition}

In the rest of this section we will assume that $G$ is a directed $N$-cycle. We have $N_c=N$ in this case, and Theorem~\ref{thm:minrank_N-1} implies that for any valid scalar coding scheme with rate $N-1$, $\ravg \geq 2(N-1)/N$.
This implies that if the vector of receiver localities $\p{r}$ is feasible for $G$ using scalar linear codes at rate $N-1$, then $r_1 + \cdots + r_N \geq 2(N-1)$.
We also know that each $r_i \geq 1$, $i \in [N]$. We now show that these two necessary conditions are also sufficient for the feasibility of $\p{r}$. 
The rest of this section is the proof of the sufficiency part of 

\begin{tcolorbox}[colback=white,boxsep=3pt,left=2pt,right=2pt,top=2pt,bottom=2pt]
\begin{theorem} \label{thm:cycle_feasible}
Let $G$ be a directed $N$-cycle and $\Fb_q$ be any finite field. A vector of receiver localities $\p{r}=(r_1,\dots,r_N)$ consisting of strictly positive integers is feasible for $G$ at rate $N-1$ through scalar linear coding over $\Fb_q$ if and only if $\sum_{i \in [N]}r_i \geq 2(N-1)$. 
\end{theorem}
\end{tcolorbox}

We will assume without loss of generality that $\sum_{i \in [N]} r_i = 2(N-1)$ and that the components of $\p{r}$ are in ascending order. We might have to relabel the vertices of $G$ for this latter assumption to hold. In general, we will assume $r_1 \leq \cdots \leq r_N$, and that $\rx_i$ demands $x_i$ and knows the message $x_{\pi(i)}$ as side information, where $\pi$ is a permutation on $[N]$ and $\pi(i) \neq i$ for any $i \in [N]$. 

Our proof for Theorem~\ref{thm:cycle_feasible} is constructive. Given a vector $\p{r}$ satisfying the conditions in the theorem, we construct a scalar linear code of rate $N-1$ that achieves receiver localities $r_1,\dots,r_N$. Towards this, we first construct a bipartite graph $\Bc$ satisfying certain adjacency properties, and then design the index code using $\Bc$.

\subsubsection*{The Bipartite Graph $\Bc$}

The bipartite graph $\Bc$ consists of the vertex sets $\{u_1,\dots,u_N\}$ representing the receivers and $\{v_1,\dots,v_{N-1}\}$ representing the codeword symbols of the length $N-1$ index code.
The edges between the two vertex sets denote the queries made by the receivers to decode their demands.
Thus, designing $\Bc$ is identical to designing the sets $R_1,\dots,R_N$.
An edge $\{u_i,v_k\}$ exists if and only if $k \in R_i$, i.e., the $k^\tth$ codeword symbol is observed by $\rx_i$.
The degree of $u_i$ is $r_i = |R_i|$, and the degree of $v_k$ is the number of receivers querying the $k^\tth$ coded symbol.
\begin{lemma} \label{lem:construction_Bc}
Let $N \geq 2$ and $\p{r}=(r_1,\dots,r_N)$ be a vector of strictly positive integers satisfying $\sum_i r_i = 2(N-1)$ and $r_1 \leq r_2 \leq \cdots \leq r_N$. There exists a bipartite graph $\Bc$ with vertex set $\{u_1,\dots,u_N\} \cup \{v_1,\dots,v_{N-1}\}$ such that\\
\emph{(i)}~the degree of $u_i$ is $r_i$, for $i \in [N]$;\\
\emph{(ii)}~the degrees of $v_1,\dots,v_{N-1}$ are all equal to $2$; and \\
\emph{(iii)}~for $i \in [N-1]$, $\{u_i,v_i\}$ is an edge and the neighbors of $u_i$ are a subset of $\{v_1,\dots,v_i\}$.
\end{lemma}
\begin{proof}
We provide a constructive proof of this lemma in Appendix~\ref{app:construct_bipartite}. Note that property~\emph{(iii)} in this lemma is equivalent to $i \in R_i$ and $R_i \subseteq [i]$ for all $i \in [N-1]$.
\end{proof}

The neighbors of $u_i$ correspond to the set $R_i \subseteq [N-1]$.
We will represent the neighbors of $v_k$ by the set $S_k \subset [N]$, i.e., $S_k=\{i \in [N]~|~k \in R_i\}$. Note that $|S_k|=2$ for $k \in [N-1]$.
For each edge $\{u_i,v_k\}$, we associate a weight $w_{i,k} \in \Fb_q$ as follows. 
If the characteristic of $\Fb_q$ is two, all the weights are equal to $1$.
Otherwise, for each $k \in [N-1]$, the two edges incident on $v_k$ are assigned the two distinct weights $1$ and $-1$, respectively, in an arbitrary fashion. 
Such an assignment of weights $w_{i,k}$ implies, over any field $\Fb_q$,
\begin{equation} \label{eq:w_ik_sum_is_zero}
\sum_{i \in S_k} w_{i,k} = 1 - 1 = 0, \text{ for all } k \in [N-1]. 
\end{equation} 

\begin{example} \label{ex:feasibility:1}
\begin{figure}[!t]
\centering
\includegraphics[width=3.4in]{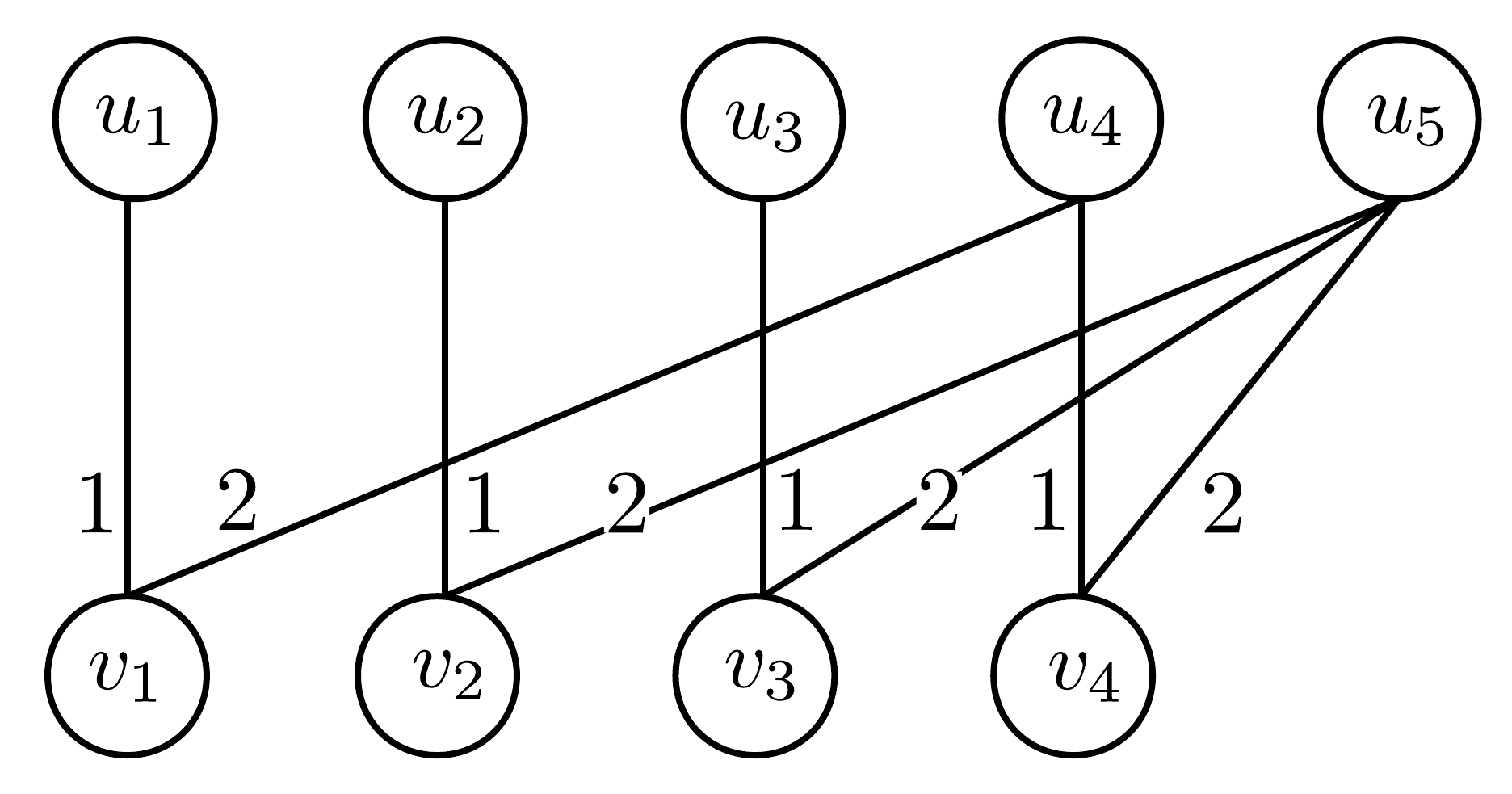}
\caption{A bipartite graph $\Bc$ for the index coding problem of Example~\ref{ex:feasibility:1}. The edges of $\Bc$ are weighted using scalars from $\Fb_3$, with weights equal to either $1$ or $2=-1$.}
\label{fig:bipartite_graph}
\end{figure} 
Consider a directed cycle of length $N=5$, where the side-informations of the receivers $\rx_1$, $\rx_2$, $\rx_3$, $\rx_4$, $\rx_5$ are the messages $x_3,x_1,x_5,x_2,x_4$, respectively. 
Note that $(\pi(1),\pi(2),\dots,\pi(5))=(3,1,5,2,4)$.
Here, $G$ is the cycle $(1,3,5,4,2,1)$. Let the desired locality vector be $\pmb{r}=(r_1,r_2,r_3,r_4,r_5) = (1,1,1,2,3)$. Let us consider code construction over the field $\Fb_3=\{0,1,2\}$, where $-1=2$. 
A bipartite graph $\Bc$ corresponding to this problem is shown in Fig.~\ref{fig:bipartite_graph}.
Note that the degrees of $u_1,\dots,u_5$ are $1,1,1,2,3$, respectively. The degrees of $v_1,\dots,v_4$ are equal to $2$. For each $i \leq 4$, $\{u_i,v_i\}$ is an edge and the neighbors of $u_i$ are a subset of $\{1,2,\dots,i\}$. For each $k \leq 4$, the weights of the two edges incident on $v_k$ are distinct, and are equal to $1$ and $2$, respectively.
\end{example}

\subsubsection*{Designing the Scalar Linear Index Code}

Let $\Fb_q$ be any finite field.
We will use the following fitting matrix $\p{A}$ to design the index code over $\Fb_q$. 
The $i^\tth$ column $\p{A}_i$ of the matrix $\p{A}$ contains exactly two non-zero components: the entry in position $i$ is $1$, and the entry in position $\pi(i)$ is $-1$. 
Note that each row of $\p{A}$ has exactly two non-zero entries with corresponding values $1$ and $-1$. Hence we have 
\begin{equation} \label{eq:sum_A_i_is_zero}
\p{A}_1 + \cdots + \p{A}_N=\p{0}.
\end{equation} 

We now construct an encoding matrix $\p{L}$ with $N-1$ columns $\p{L}_1,\dots,\p{L}_{N-1}$ corresponding to the 
fitting matrix $\p{A}$ and for the receiver queries $R_1,\dots,R_N$ dictated by $\Bc$.
A sufficient condition for $\p{L}$ to satisfy these properties is
\begin{align} \label{eq:weights_L_and_A}
\p{A}_i = \sum_{k \in R_i} w_{i,k} \p{L}_k, \text{ for } i \in [N],
\end{align}
since this implies $x_i - x_{\pi(i)} = \xb^\tr\p{A}_i$ can be obtained from the coded symbols $\xb^\tr\p{L}_k$, $k \in R_i$.
We prove the existence of $\p{L}$ by viewing~\eqref{eq:weights_L_and_A} as a set of $N$ linear equations in unknowns $\p{L}_1,\dots,\p{L}_{N-1}$, and showing that this system has a unique solution.
Towards this, we claim that the last equation in this system is linearly dependent on the first $N-1$ equations, and hence, is redundant. 
The proof for this claim follows from using~\eqref{eq:w_ik_sum_is_zero} and~\eqref{eq:sum_A_i_is_zero} as below
\begin{align*}
\p{0} &= \p{A}_1 + \cdots + \p{A}_N - \sum_{k \in [N-1]} \,\sum_{i \in S_k} w_{i,k} \p{L}_k \\
&\qquad \qquad (\text{last term corresponds to all the edges in } \Bc)\\
      &= \p{A}_1 + \cdots + \p{A}_N - \sum_{i \in [N]} \sum_{k \in R_i} w_{i,k} \p{L}_k \\
&\qquad \qquad (\text{alternative way to enumerate the edges in } \Bc)\\
      &= \sum_{i \in [N-1]} \! ( \p{A}_i -  \sum_{k \in R_i} w_{i,k} \p{L}_k ) \, + \, ( \p{A}_N -  \sum_{k \in R_N} w_{N,k} \p{L}_k ). \\
\end{align*} 
We conclude that~\eqref{eq:weights_L_and_A} is equivalent to the following system consisting of $N-1$ equations with an equal number of unknowns
\begin{align} \label{eq:weights_L_and_A_2}
\p{A}_i = \sum_{k \in R_i} w_{i,k} \p{L}_k, \text{ for } i \in [N-1].
\end{align}
From Lemma~\ref{lem:construction_Bc} we know that $i \in R_i$ and $R_i \subseteq [i]$. Thus, the unknowns involved in the $i^\tth$ equation of the linear system~\eqref{eq:weights_L_and_A_2} are a subset of $\{\p{L}_1,\dots,\p{L}_i\}$ and necessarily include the unknown $\p{L}_i$. Hence,~\eqref{eq:weights_L_and_A_2} is a triangular system of linear equations, and therefore, has a unique solution.
We conclude that the locality vector $\p{r}$ is feasible.

\begin{example}
We will continue with the index coding problem of Example~\ref{ex:feasibility:1}. The fitting matrix for this problem is
\begin{equation*}
 \pmb{A} = \begin{bmatrix}
           1 & 2 & 0 & 0 & 0 \\
           0 & 1 & 0 & 2 & 0 \\
           2 & 0 & 1 & 0 & 0 \\
           0 & 0 & 0 & 1 & 2 \\
           0 & 0 & 2 & 0 & 1
           \end{bmatrix}.
\end{equation*}
We now determine the encoder matrix $\pmb{L}$, that consists of $4$ columns $\pmb{L}_1,\dots,\pmb{L}_4$, using the bipartite graph $\Bc$ in Fig~\ref{fig:bipartite_graph}. The system of linear equations~\eqref{eq:weights_L_and_A_2} for this problem is
\begin{equation*}
\pmb{A}_1 = \pmb{L}_1,~\pmb{A}_2=\pmb{L}_2,~\pmb{A}_3=\pmb{L}_3, \text{ and } \pmb{A}_4 = 2\pmb{L}_1 + \pmb{L}_4.
\end{equation*} 
The unique solution to this system is 
\begin{equation*}
\pmb{L} = [\pmb{L}_1 \cdots \pmb{L}_4]
=[\pmb{A}_1~\pmb{A}_2~\pmb{A}_3~\pmb{A}_4+\pmb{A}_1]
=\begin{bmatrix}
1 & 2 & 0 & 1 \\
0 & 1 & 0 & 2 \\
2 & 0 & 1 & 2 \\
0 & 0 & 0 & 1 \\
0 & 0 & 2 & 0
\end{bmatrix} \!\!.
\end{equation*} 
With this encoding matrix, the transmitted codeword is 
\begin{align*}
\pmb{c}&=(c_1,c_2,c_3,c_4) \\
       &=(x_1-x_3, x_2-x_1,x_3-x_5,x_1-x_2-x_3+x_4).
\end{align*} 
The receivers $\rx_1,\rx_2,\dots,\rx_5$ decode their demanded messages using a subset of codeword symbols and their own side information as follows
\begin{align*}
x_1 &= c_1 + x_3, ~~~~~~ x_2 = c_2 + x_1, ~~~~~~x_3=c_3 + x_5, \\
x_4 &= c_4 - c_1 + x_2 ~~\text{ and }~~ x_5 = -c_2-c_3-c_4+x_4.
\end{align*}
The localities at the five receivers are $(r_1,r_2,\dots,r_5)=(1,1,1,2,3)$, which are as stipulated in Example~\ref{ex:feasibility:1}. 
\end{example}

\section{Constructing Locally Decodable\\ Index Codes}  \label{sec:design}

In the previous sections we considered index codes for specific cases: locality is equal to one, side information graph is a directed cycle, and minrank is one less than the order of the problem. 
We developed matching converses for each of these cases and identified optimal index codes.
In this section we propose techniques to construct linear locally decodable index codes for arbitrary index coding problems. We also show that some of the traditional index code constructions can be modified to yield locally decodable codes.

Recall the decoding technique for scalar linear codes from Section~\ref{sec:sub:scalar_structure}.
Let the encoding matrix be $\p{L}$, the corresponding fitting matrix be $\p{A}$ and the $i^\tth$ column of $\p{A}$ be $\p{A}_i=\p{e}_i + \p{u}_i$ where $\p{u}_i \lhd K_i$.
Receiver $\rx_i$ decodes its demand by computing a specific linear combination of codeword symbols and its side information. The linear combination $\cb^\tr\p{d}_i$, $\p{d}_i \in \Fb_q^\ell$, of the codeword symbols is chosen such $\p{L}\p{d}_i = \p{A}_i$, and hence,
\begin{equation}
\cb^\tr\p{d}_i = \xb^\tr\p{L}\p{d}_i = \xb^\tr\p{A}_i = \xb^\tr(\p{e}_i + \p{u}_i).
\end{equation}  
With this choice of linear combination, the demand $x_i$ can be decoded as 
$x_i = \cb^\tr\p{d}_i - \xb^\tr\p{u}_i$,
Note that $R_i=\supp(\p{d}_i)$ since the receiver needs to observe only those components of $\cb$ that correspond to the non-zero entries of $\p{d}_i$. 
Thus, the locality of the $i^\tth$ receiver equals the Hamming weight of $\p{d}_i$, i.e., $r_i = \wt(\p{d}_i)$.
If $\p{A}_i$ is one of the columns of the encoding matrix $\p{L}$, then $\p{d}_i$ can be chosen such that $\wt(\p{d}_i)=1$ resulting in $r_i=1$. If $\p{L}$ does not contain $\p{A}_i$ as one of its columns, then we have the naive upper bound $r_i = \wt(\p{d}_i) \leq \ell$.

\subsection{Separation Based Scheme using Covering Codes}

We can arrive at a locally decodable code using a separation based technique where the encoder matrix $\p{L}=\p{L}'\p{H}$ is the product of two matrices: an optimal index coding matrix $\p{L}'$ with number of columns equal to $\minrank(G)$, and the parity-check matrix $\p{H}$ of a \emph{covering code} $\mathscr{C}$ with length $\ell$ and dimension $\ell-\minrank(G)$. 
The linear code $\mathscr{C}$ is chosen such that its covering radius is equal to the desired locality $r$, i.e., Hamming spheres of radius $r$ centered around the codewords of $\mathscr{C}$ cover the entire space $\Fb_q^\ell$. 
Among all covering codes over $\Fb_q$ with covering radius $r$ and co-dimension $\minrank(G)$ we choose the one with the smallest possible blocklength $\ell$.
Note that the number of rows of $\p{H}$ is $\minrank(G)$.
Since the covering radius of $\mathscr{C}$ is $r$, it is well known that, any vector of length $\minrank(G)$ is a linear combination of at the most $r$ columns of $\p{H}$.
If $\p{A}$ is a fitting matrix corresponding to $\p{L}'$, every column of $\p{A}$ is a linear combination of columns of $\p{L}'$.
Thus, $\p{A}_i$ can be expressed as a linear combination of at the most $r$ columns of the matrix \mbox{$\p{L}=\p{L}'\p{H}$}. 
We conclude that $\p{L}$ is a valid scalar linear encoder matrix for $G$ with locality $r$.

\subsection{Codes using Acyclic Induced Subgraph Covers of $G$} \label{sec:ais_cover}

We will design vector linear codes using the acyclic induced subgraphs of $G$. 
First, we require the following result.

Note that for any subset $S \subset [N]$ of vertices of $G$, $G_S$ denotes the subgraph of $G$ induced by $S$.

\begin{lemma} \label{lem:ais_lemma}
Let $\ell \geq \minrank(G)$ and the subgraph $G_S$ of $G$ induced by the subset $S \subseteq [N]$ be a directed acyclic graph. There exists a scalar linear index code of length $\ell$ for $G$ such that $r_i=1$ for every $i \in S$.
\end{lemma} 
\begin{proof}
Let $\p{L} \in \Fb_q^{N \times \ell}$ be the encoding matrix of any valid scalar linear code for $G$, and $\p{A}=[\,\p{A}_1~\cdots~\p{A}_N \,]$ be the fitting matrix corresponding to $\p{L}$.
Since $G_S$ is acyclic, there exists a \emph{topological ordering} $i_1,\dots,i_{|S|}$ of its vertex set $S=\{i_1,\dots,i_{|S|}\}$, i.e., for any $1 \leq a < b \leq |S|$, there exists no directed edge $(i_b,i_a)$ in $G_S$, and hence, $G$ does not contain the edge $(i_b,i_a)$. 
It follows that for any choice of $1 \leq a < b \leq |S|$, $i_a \notin K_{i_b}$, and hence, \mbox{$\p{A}_{i_a,i_b}=0$}. Further, for any $1 \leq a \leq |S|$, $\p{A}_{i_a,i_a}=1$ since the diagonal entries of $\p{A}$ are equal to $1$.
Let $\p{E}$ be the $|S| \times |S|$ square submatrix of $\p{A}$ composed of the rows and columns indexed by $S$. It follows that if the rows and columns of $\p{E}$ are permuted according to the topological ordering $i_1,\dots,i_{|S|}$, then $\p{E}$ is lower triangular and all the entries on its main diagonal are equal to $1$. Thus $\p{E}$ is a full-rank matrix, and hence, the columns $\p{A}_{i_1},\dots,\p{A}_{i_{|S|}}$ of $\p{A}$ are linearly independent.

Consider the matrix $\p{L}' \in \Fb_q^{N \times \ell}$ constructed as follows. 
Let the first $|S|$ columns of $\p{L}'$ be $\p{A}_{i_1},\dots,\p{A}_{i_{|S|}}$. 
The remaining $\ell-|S|$ columns of $\p{L}'$ are chosen from among the columns of $\p{L}$ such that the column spaces of $\p{L}'$ and $\p{L}$ are identical. 
This is possible since $\p{A}_{i_1},\dots,\p{A}_{i_{|S|}}$ are linearly independent and are contained in the column space of $\p{L}$.
By construction, the column space of $\p{L}'$ contains the column space of $\p{A}$, and hence, $\p{L}'$ is a valid scalar linear index coding matrix for $G$. Also, for any $i \in S$, $\p{A}_i$ is a column of $\p{L}'$, and hence, $r_i=1$.
\end{proof}

\begin{definition} 
A set of $P$ subsets $S_1,\dots,S_P \subseteq [N]$ of the vertex set of the side information graph $G$ is a \emph{$Q$-fold acyclic induced subgraph (AIS) cover} of $G$ if \emph{(i)}~$S_1 \cup \cdots \cup S_P = [N]$, \emph{(ii)}~each $i \in [N]$ is an element of at least $Q$ of the $P$ subsets $S_1,\dots,S_P$, and \emph{(iii)}~all the $P$ induced subgraphs $G_{S_1},\dots,G_{S_M}$ are acyclic.
\end{definition}

Given an AIS cover $S_1,\dots,S_P$ of $G$ and any \mbox{$\ell \geq \minrank(G)$},
we construct a vector linear code as follows. 
From Lemma~\ref{lem:ais_lemma}, we know that for each $j \in [P]$, there exists a valid scalar linear encoding matrix, say $\p{L}^{(j)}$, with codelength $\ell$ such that the locality of every receiver $i \in S_j$ is $1$. 
Consider a vector linear index code that encodes $P$ independent instances of the scalar messages $x_1,\dots,x_N \in \Fb_q$ using the encoding matrices $\p{L}^{(1)},\dots,\p{L}^{(P)}$, respectively. 
The broadcast rate of this scheme is $\ell$. 
If $S_1,\dots,S_P$ is a $Q$-fold AIS cover of $G$, for each \mbox{$i \in [N]$}, there exist $Q$ scalar linear encoders among $\p{L}^{(1)},\dots,\p{L}^{(P)}$ that provide locality $1$ at the $i^\tth$ receiver. 
The locality provided by the remaining $(P-Q)$ encoders is at most $\ell$ at this receiver. 
Thus the number of encoded symbols queried by any receiver in this vector linear coding scheme is at the most $Q + (P-Q)\ell$. 
Normalizing this by the number of message instances $P$, we observe that the locality of this scheme is at the most ${(Q + (P-Q)\ell)}/{P}$. 
Thus, we have proved 
\begin{theorem}~\label{thm:ais_cover}
If there exists a $Q$-fold AIS cover of $G$ consisting of $P$ subsets of its vertex set, and if $\ell \geq \minrank(G)$,
there exists a vector linear code for $G$ with broadcast rate $\ell$, message length $M=P$, and locality at the most $({Q + (P-Q)\ell})/{P}$, and hence,
\begin{equation*}
\beta_{G,q}^*\left( \frac{Q + (P-Q)\ell}{P} \right) \leq \ell.
\end{equation*}  
\end{theorem} 

As an application of Theorem~\ref{thm:ais_cover}, consider the following coding scheme. 
Let the parameter \mbox{$t \geq 1$} be such that the side information graph $G$ contains no cycles of length $t$ or less. 
With $P=\binom{N}{t}$, let $S_1,\dots,S_P$ be the set of all subsets of $[N]$ of size $t$. 
The subgraph of $G$ induced by $S_j$, for any $j \in [P]$, is acyclic since $|S_j|=t$. Further, each $i \in [P]$ is an element of 
$Q = \binom{N-1}{t-1}$ 
of the $P$ subsets $S_1,\dots,S_P$. 
Hence the resulting locality is $({Q + (P-Q)\ell})/{P}$ which can easily be shown to be equal to $(t+(N-t)\ell)/{N}$. Hence, we have 

\begin{corollary}~\label{cor:ais_cover_t}
If $G$ contains no cycles of length $t$ or less and $\ell \geq \minrank(G)$,
then we can achieve broadcast rate $\ell$ with locality at the most $(t+(N-t)\ell)/{N}$.
\end{corollary}

Using $t=1$ in Corollary~\ref{cor:ais_cover_t} we immediately arrive at the following result.

\begin{corollary} \label{cor:lem:arbitrary_locality}
Let $G$ be any index coding problem and $\ell \geq \minrank(G)$.
There exists a vector linear index code for $G$ with message length $M=N$, broadcast rate $\beta=\ell$ and locality $r \leq (1+(N-1)\ell)/{N}$.
\end{corollary}

\subsection{Side Information Graphs with Circular Symmetry} \label{sec:minrank_const}

We will now construct vector linear index codes for side information graphs $G$ that satisfy a symmetry property. 
Consider the permutation $\sigma$ on the set $[N]$ that maps $i \in [N]$ to $\sigma(i)=(i \mod N) + 1$, i.e., $\sigma(1)=2, \sigma(2)=3,\dots,\sigma(N)=1$. In this subsection we will assume $G$ to be any directed graph with vertex set $[N]$ such that $\sigma$ is an automorphism of $G$. 
Such unicast index coding problems have been considered before in the literature, see~\cite{MCJ_IT_14}, and are related to topological interference management~\cite{Jaf_IT_14}.
The following theorem shows that any rate $\ell \geq \minrank(G)$ can be achieved using a vector linear code with $r$ at the most $\ell(N-\ell+1)/N$.

\begin{theorem} \label{thm:cyclic_const}
If the cyclic permutation $\sigma$ is an automorphism of $G$ and if $\ell \geq \minrank(G)$, then
\begin{equation*}
 \beta_{G,q}^*\left( \frac{\ell(N-\ell+1)}{N} \right) \leq \ell.
\end{equation*} 
\end{theorem}
\begin{proof}
Since $\ell \geq \minrank(G)$, there exists an $\p{A} \in \Fb_q^{N \times N}$ that fits $G$ and which is of rank $\ell$.
Let $\p{L}$ be an $N \times \ell$ matrix composed of a set of $\ell$ linearly independent columns of $\p{A}$.
Note that when $\p{L}$ is used as a scalar linear index code there exist $\ell$ receivers with locality $1$ since exactly $\ell$ columns of $\p{A}$ appear as columns of $\p{L}$. 

Let $\p{P}$ be the permutation matrix obtained by cyclically shifting down the rows of the $N \times N$ identity matrix by one position. It is straightforward to verify that $\p{PAP}^\tr$ fits the graph $\sigma(G)$. 
Since $\sigma$ is an automorphism of $G$, we have $\sigma(G)=G$, and hence, $\p{PAP}^\tr$ fits $G$. 
Since the column space of $\p{PL}$ is identical to that of $\p{PAP}^\tr$, $\p{PL}$ represents a valid scalar linear index code for $G$. 
Extending this argument we deduce that for any $j \in [N]$, the matrix $\p{A}^{(j)}=\p{P}^j\p{A}(\p{P}^j)^\tr$ fits $G$, and the column space of $\p{L}^{(j)}=\p{P}^j\p{L}$ is identical to that of $\p{A}^{(j)}$. 
The matrices $\p{L}^{(j)}$ and $\p{A}^{(j)}$, $j \in [N]$, represent $N$ valid scalar linear codes and the corresponding fitting matrices for $G$. 
The receiver $\rx_i$ has locality $1$ when using the $j^\tth$ code $\p{L}^{(j)}$ if the $i^\tth$ column of $\p{A}^{(j)}$ appears as a column of $\p{L}^{(j)}$, else its locality is at the most $\ell$.
Using the fact that the group $\{\sigma,\sigma^2,\dots,\sigma^N=1\}$ acts transitively on the vertices of $G$, we observe that for any $i \in [N]$, there exist $\ell$ distinct values of $j$ such that the $i^{\text{th}}$ column of $\p{A}^{(j)}$ is a column of $\p{L}^{(j)}$. 
Thus, for any $i \in [N]$, the locality of $\rx_i$ is $1$ for $\ell$ of the $N$ codes $\p{L}^{(1)},\dots,\p{L}^{(N)}$, and is at the most $\ell$ for the remaining $N-\ell$ codes. Time-sharing these $N$ codes we obtain a vector linear code with message length $M=N$, rate $\ell$ and locality at the most $(\ell + \ell(N-\ell))/N = \ell(N-\ell+1)/N$ for each of the $N$ receivers.
\end{proof}

\subsection{Codes from Optimal Coverings of $G$}

Several index coding schemes in the literature partition the given problem (side information graph $G$) into subproblems (subgraphs of $G$), and apply a pre-defined coding technique on each of these subproblems independently. 
The overall codelength is the sum of the codelengths of the individual subproblems.
The broadcast rate is then reduced by optimizing over all possible partitions of $G$. We will now quickly recall a few such \emph{covering-based} schemes, and then show that they can be modified to guarantee locality.


\subsubsection*{Partition Multicast or Partial Clique Covering~\cite{BiK_INFOCOM_98,BKL_IT_13,TDN_ISIT_12}}
Let $G_S$ be the subgraph of $G$ induced by the subset of vertices \mbox{$S \subseteq [N]$}. 
The number of information symbols in the index coding problem $G_S$ is $|S|$, and the side information of receiver \mbox{$i \in S$} in $G_S$ is $K_i \cap S$. 
The partition multicast scheme uses a scalar linear encoder for $G_S$ where the encoding matrix is the transpose of the parity-check matrix of an MDS code of length $|S|$ and dimension $\min_{i \in S}|K_i \cap S|$. This code for $G_S$ encodes messages of length \mbox{$m_S=1$} and has codelength $\ell_S=|S|-\min_{i \in S} |K_i \cap S|$. We will use the trivial value of overall locality $r_S=\ell_S$ for this coding scheme.
The finite field $\Fb_q$ must be sufficiently large for the said MDS codes to exist. 

\subsubsection*{Cycle Covering~\cite{NTZ_IT_13,CASL_ISIT_11}}

If $G_S$ is a directed cycle of length $|S|$, then it is encoded using the scalar linear coding scheme described in Example~\ref{ex:simple_code_cycle}.
resulting in message length \mbox{$m_S=1$}, index codelength $\ell_S=|S|-1$, and locality $r_S=2$.
If $G_S$ is not a directed cycle, then its information symbols are transmitted uncoded resulting in $m_S=1$, $\ell_S=|S|$ and locality $r_S=1$. 

In similar vein to partition multicast and cycle covering schemes, consider the following proposed technique that applies the optimal scalar linear code over each subgraph $G_S$. 

\subsubsection*{Minrank Covering}
Encode each subgraph $G_S$ using its own optimal scalar linear index code. The message length $m_S=1$, codelength $\ell_S$ equals $\minrank(G_S)$, and use trivial value for locality $r_S=\minrank(G_S)$. By partitioning $G$ into subgraphs $G_S$ of small minrank we can achieve a small locality for the overall scheme.

Now consider any covering-based index coding technique (such as partition multicast, cycle covering or minrank covering) for $G$. 
Let the scalar linear index code associated with the subgraph $G_S$, $S \subseteq [N]$, have codelength $\ell_S$ and locality $r_S$. 
The overall index code uses a partition of the vertex set $[N]$, which is represented by the tuple $(a_S, S \subset [N])$, where each $a_S \in \{0,1\}$ is such that the partition of $[N]$ consists of all subsets $S$ with $a_S=1$.
Note that $(a_S, S \subset [N])$ represents a partition of $[N]$ if and only if $\sum_{S: i \in S}a_S=1$ for every $i \in [N]$, i.e., every vertex $i$ is contained in exactly one of the subsets in the partition.
The covering-based index coding technique applies an index code of length $\ell_S$ and locality $r_S$ to each subgraph $G_S$ with \mbox{$a_S=1$} independently. Thus the codelength of the overall index code is \mbox{$\ell=\sum_{S: a_S=1}\ell_S=\sum_{S \subseteq [N]}a_S\ell_S$} and locality is $r=\max_{S:a_S=1}r_S=\max_{S \subseteq [N]}a_Sr_S$.
By optimizing over all possible partitions of $G$, we have 

\begin{theorem}[Covering with locality] \label{thm:partition_scalar}
Consider a family of scalar linear index codes, one for each $G_S$, $S \subseteq [N]$, with length $\ell_S$ and locality $r_S$. 
Given any $r \geq 1$, there exists a scalar linear code for $G$ with locality at the most $r$ and rate equal to the solution of the following integer program:
minimize $\sum_{S \subseteq [N]}a_S\ell_S$ subject to the following constraints, 
$\sum_{S: i \in S}a_S = 1$, for all $i \in [N]$, and 
$a_Sr_S \leq r$, where $a_S \in \{0,1\}$ for all $S \subseteq [N]$.
\end{theorem}
The second constraint \mbox{$a_Sr_S \leq r$} in Theorem~\ref{thm:partition_scalar} ensures that the locality of the resulting coding scheme is at the most $r$. Since $a_S \in \{0,1\}$, this implies that when solving for the optimal partition, the integer program considers only those subsets $S$ with $r_S \leq r$, i.e., locality $r$ is achieved by partitioning $G$ into subproblems of small locality.

The above technique can be extended to vector linear codes by using a linear programming relaxation to allow $0 \leq a_S \leq 1$. 
The subproblems $G_S$ can themselves be encoded using vector linear codes of small locality, as in
Theorems~\ref{thm:ais_cover} and~\ref{thm:cyclic_const},
say with rate $\beta_S$ and locality $r_S$. 
Then the overall achievable rate by time-sharing over these subproblems, for a given locality $r$, is the solution to $\min \sum_{S \subseteq [N]} a_S\beta_S$ subject to $\sum_{S:i\in S}a_Sr_S \leq r$ and $\sum_{S:i \in S}a_S=1$ for each $i \in [N]$, where $a_S \in [0,1]$ for all $S \subseteq [N]$.

\section{Discussion} \label{sec:conclusion}

We introduced the problem of designing index codes that are locally decodable and have identified several techniques to construct such codes.
We identified the optimal trade-off between rate and locality for the following cases: when locality $r=1$, side information graph is the directed $3$-cycle, vector linear codes for directed cycles of any length, and scalar linear codes when minrank is one less than the number of messages.

\subsection{Relation to Overcomplete Dictionaries}

Locally decodable index codes seem to be related to the problem of sparse representation of sets of vectors.
A scalar linear index code with locality $r$ is characterized by a valid encoder matrix $\p{L}$ with a corresponding fitting matrix $\p{A}$ such that any column of $\p{A}$ is some linear combination of at the most $r$ columns of $\p{L}$. Thus the columns of $\p{L}$ serve as an overcomplete basis for a sparse representation of the columns of $\p{A}$. 
Given an index coding problem, we must design a fitting matrix $\p{A}$ and a corresponding overcomplete basis $\p{L}$ for $\p{A}$ such that both the locality $r$ and the codelength $\ell$ are small.

\subsection{On Generalization of Results on Locality-Rate Trade-off}

Stronger achievability and/or converse results may be required to gain deeper insights into the locality-rate trade-off of a general index coding problem. 
In Section~\ref{sec:minrank_N_minus_1} we identified optimal scalar linear codes for side information graphs $G$ with $\minrank(G) = N-1$. 
The key property used in Section~VI is that any such directed graph $G$ contains at least one directed cycle (see Lemma~\ref{lem:GS_contains_cycle} and Theorem~\ref{thm:minrank_N-1}). 
A natural direction to generalize this result is to consider side information graphs $G$ with smaller values of minrank, such as \mbox{$\minrank(G) \leq N-2$}.
However, the structural properties of directed graphs with \mbox{$\minrank(G)=N-2$} are not as extensively understood in the index coding community as the more extreme (and simpler) case of $\minrank(G)=N-1$.
For instance, reference~\cite{Ong_Netcod_14} identifies a subset of index coding problems whose minrank equals $N-2$.

Another avenue of work is the generalization of the results on directed $3$-cycle (in Section~\ref{sec:3cycle}) to directed cycles of length $N>3$.
The case $N=3$ is highly structured and this facilitated our analysis in Section~\ref{sec:3cycle}. However, some of this symmetry is lost when $N$ takes values larger than $3$ and this imposes analytical difficulties (see Remark~\ref{rem:3cycle}). 

One of the extreme operating points in the trade-off between locality $r$ and rate $\beta$ is the point where the broadcast rate is optimal, i.e., minimum possible with no restriction imposed by locality. 
Thus, when deriving the optimal trade-off between $r$ and $\beta$ we are naturally restricted to the family of index coding problems for which the minimum broadcast rate or capacity is explicitly known.
However, determining the capacity of index coding problems is difficult in general, and the capacity is known for only some families of side information graphs.



\appendices


\section{Invariance of Locality-Rate Trade-off with respect to Channel and Message Alphabets} \label{app:invariance_alphabet}

The channel model considered in this paper assumes that the message alphabet and the channel alphabet are identical, i.e., all the messages and the codeword are vectors over the same alphabet $\Ac$. 
The definition of rate $\beta$, locality $r$ and the optimal locality-rate trade-off $\beta_G^*(r)$ given in Section~\ref{sec:system_model} are specific to this model. 
In this section, we consider the general scenario where the channel alphabet, denoted by $\Zc$, and the message alphabet $\Ac$ can be different. 
We will define locality when $\Ac \neq \Zc$, and show that the optimal locality-rate trade-off (that captures non-linear codes as well) is independent of the choice of $\Ac$ and $\Zc$ as long as these sets are finite and of size at least $2$.
The results of this appendix are summarized in Appendix~\ref{app:sec:summary_invariance}.

\subsection{The General System Model}

As in Section~\ref{sec:system_model}, we will assume that $G$ is a given side information graph with $N$ vertices, and the messages $\p{x}_1,\dots,\p{x}_N$ are elements of $\Ac^M$, where $\Ac$ is the message alphabet and $M$ is the message length. The codeword $\p{c}$ is a vector of length $\ell$ over the channel alphabet $\Zc$, i.e., $\p{c} = (c_1,\dots,c_{\ell}) \in \Zc^{\ell}$. Both $\Ac$ and $\Zc$ are finite sets.
The receiver $\rx_i$ observes $\p{c}_{R_i} = (c_k, k \in R_i) \in \Zc^{|R_i|}$, where $R_i \subseteq [\ell]$, and decodes its demand $\p{x}_i$ using $\p{c}_{R_i}$ and its side information $\p{x}_j$, $j \in K_i$.
The rate of this code is defined as
\begin{equation*}
 \beta = \frac{\log_2 |\Zc|^\ell}{\log_2 |\Ac|^M} = \frac{\ell}{M} \times \frac{\log_2 |\Zc|}{\log_2 |\Ac|}.
\end{equation*} 
This is consistent with the definitions of rate used in the literature; for instance, see~\cite{BKL_IT_13}. We define the locality of receiver $\rx_i$ as
\begin{equation*}
 r_i = \frac{\log_2 |\Zc|^{|R_i|}}{\log_2 |\Ac|^M} = \frac{|R_i|}{M} \times \frac{\log_2 |\Zc|}{\log_2 |\Ac|}.
\end{equation*} 
Here, $\log_2 |\Zc|^{|R_i|}$ measures the number of bits observed from the channel and $\log_2 |\Ac|^M$ is the information content of the desired message (in bits). Thus $r_i$ is the number of transmission bits queried per decoded message bit at $\rx_i$. 
The average locality $\ravg$ and overall locality $r$ are the arithmetic average and the maximum of $r_1,\dots,r_N$, respectively.
When $\Zc=\Ac$ these definitions of $\beta$, $r_i$, $\ravg$ and $r$ coincide with the ones used in the main text of this paper.

\begin{definition}
The \emph{optimal broadcast rate} $\beta_G^*(r,\Ac,\Zc)$ is the infimum among the broadcast rates $\beta$ of all valid index codes (including non-linear codes) with message alphabet $\Ac$, channel alphabet $\Zc$, and with locality at the most $r$, considering all possible message lengths $M \geq 1$.
\end{definition}

The definition of $\beta_G^*(r,\Ac,\Zc)$ captures the possible dependence (if any) of the locality-rate trade-off on the choice of message and channel alphabets. The function $\beta_G^*(r)$ in the main text is identical to $\beta_G^*(r,\Ac,\Ac)$.

\subsection{Invariance with respect to channel alphabet} \label{sec:app:invariance_channel_alph}

We now show that $\beta_G^*(r,\Ac,\Zc)$ is independent of $\Zc$. We will do so by showing that $\beta_G^*(r,\Ac,\Zc') \leq \beta_G^*(r,\Ac,\Zc)$ for any choice of alphabets $\Zc$ and $\Zc'$. Reversing the roles of $\Zc$ and $\Zc'$ in this inequality shows that $\beta_G^*(r,\Ac,\Zc) \leq \beta_G^*(r,\Ac,\Zc')$, and hence, $\beta_G^*(r,\Ac,\Zc)$ is independent of $\Zc$.

\begin{lemma} \label{lem:app:channel_alphabet}
Consider any valid index code for $G$ with message and channel alphabets $\Ac$ and $\Zc$ with rate $\beta$ and overall locality $r$, with $|\Ac|,|\Zc| \geq 2$. For any finite set $\Zc'$, with $|\Zc'| \geq 2$, and any $\epsilon > 0$, there exists a valid index code for $G$ with message alphabet $\Ac$, channel alphabet $\Zc'$, rate at the most $(1 + \epsilon)\beta$ and locality at the most $(1 + \epsilon) r$.
\end{lemma}
\begin{proof}
Let $\Cs$ be the given valid index code over the channel alphabet $\Zc$, with codelength $\ell$, message length $M$, receiver queries $R_1,\dots,R_N$. 
Denote the decoders at the $N$ receivers by $\Ds_1,\dots,\Ds_N$. We will use this index code to build a valid code over the channel alphabet $\Zc'$. To do so, we choose positive integers $t$ and $s$ such that
\begin{equation} \label{eq:app:lem:channel_alphabet}
 \frac{\log_2 |\Zc|}{\log_2 |\Zc'|} \leq \frac{t}{s} \leq (1 + \epsilon) \frac{\log_2 |\Zc|}{\log_2 |\Zc'|}.
\end{equation} 
We consider $s$ generations of messages and encode each of them, independently, using $\Cs$. For $j \in [s]$, we encode the $j^\tth$ generation $\p{x}_1^{(j)},\dots,\p{x}_N^{(j)} \in \Ac^M$ using $\Cs$ to obtain the codeword $\p{c}^{(j)}=(c_1^{(j)},\dots,c_\ell^{(j)}) \in \Zc^\ell$.
Note that the number of symbols encoded per message, considering all generations, is $M' = sM$.
Now, to design an index code $\Cs'$ for the channel alphabet $\Zc'$, we map the vectors $\p{c}^{(1)},\dots,\p{c}^{(s)}$ to a vector of length $\ell'=t\ell$ over $\Zc'$ as follows.
For each $k \in [\ell]$, we use any one-to-one map from $\Zc^s$ to ${\Zc'}^{\,t}$ to map the vector $(c_k^{(1)},\dots,c_k^{(s)})$ to $(d_k^{(1)},\dots,d_k^{(t)})$. This is possible since $|\Zc|^s \leq |\Zc'|^t$, see~\eqref{eq:app:lem:channel_alphabet}.
The codeword transmitted for the index code $\Cs'$ is the vector $(d_k^{(j)}, k \in [\ell], j \in [t])$.

To decode the $i^\tth$ message, receiver $\rx_i$ observes the symbols $(d_k^{(j)}, k \in R_i, j \in [t])$. The number of transmissions queried by $\rx_i$ is $|R_i'| = t|R_i|$. 
For each $k \in R_i$, the receiver uses $(d_k^{(1)},\dots,d_k^{(t)})$ to retrieve $(c_k^{(1)},\dots,c_k^{(s)})$. 
Thus, $\rx_i$ can compute $(c_k^{(j)}, k \in R_i)$, for every $j \in [s]$. For each $j \in [s]$, the vector $(c_k^{(j)}, k \in R_i)$ can be used with the side information available at $\rx_i$ to decode $\p{x}_i^{(j)}$ using the decoder $\Ds_i$ of the original index code $\Cs$. 
This shows that $\Cs'$ is a valid index code with channel alphabet $\Zc'$.

Using~\eqref{eq:app:lem:channel_alphabet} we observe that the rate $\beta'$ of $\Cs'$ satisfies
\begin{align*}
\beta' = \frac{\ell' \log_2 |\Zc'|}{M' \log_2 |\Ac|} = \frac{t \ell \, \log_2 |\Zc'|}{s M \, \log_2 |\Ac|} &\leq (1 + \epsilon) \frac{\ell \log_2 |\Zc|}{M \log_2 |\Ac|} \\ 
&= (1 + \epsilon)\beta.
\end{align*} 
Similarly, the locality $r_i'$ of the $i^\tth$ receiver for the index code $\Cs'$ is
\begin{align*}
r_i' = \frac{|R_i'| \log_2 |\Zc'|}{M' \log_2 |\Ac|} = \frac{t |R_i| \log_2 |\Zc'|}{s M \log_2 |\Ac|} &\leq (1 + \epsilon) \frac{|R_i| \log_2 |\Zc|}{M \log_2 |\Ac|} \\
&= (1 + \epsilon) r_i,
\end{align*} 
where $r_i$ is the receiver locality of $\rx_i$ for the code $\Cs$.
Hence, the overall locality of $\Cs'$ is at the most $(1 + \epsilon)r$.

\end{proof}

Considering vanishingly small values of $\epsilon$, and choosing $\beta$ arbitrarily close to $\beta_G^*(r,\Ac,\Zc)$, Lemma~\ref{lem:app:channel_alphabet} implies that $\beta_G^*(r,\Ac,\Zc') \leq \beta_G^*(r,\Ac,\Zc)$.

\subsection{Invariance with respect to message alphabet}

We will use arguments similar to Appendix~\ref{sec:app:invariance_channel_alph} to show that $\beta_G^*(r,\Ac,\Zc)$ is independent of $\Ac$. In particular, we show that for any choice of finite sets $\Ac$ and $\Ac'$, with $|\Ac|, |\Ac'| \geq 2$, $\beta_G^*(r,\Ac',\Zc) \leq \beta_G^*(r,\Ac,\Zc)$. From the symmetry of this result, it follows that the inequality holds in the other direction as well, implying equality.

\begin{lemma} \label{lem:app:message_alphabet}
Suppose there exists a valid index code for $G$ with message alphabet $\Ac$ and channel alphabet $\Zc$ of rate $\beta$ and locality $r$, with $|\Ac|,|\Zc| \geq 2$. Let $\Ac'$ be any finite set with size at least two. For any $\epsilon > 0$, there exists a valid index code for $G$ with message alphabet $\Ac'$, channel alphabet $\Zc$, with rate at the most $(1+\epsilon) \beta$ and locality at the most $(1 + \epsilon)r$.
\end{lemma}
\begin{proof}
Let $\Cs$ denote a valid index code with codelength $\ell$, message length $M$, receiver localities $r_1,\dots,r_N$, and receiver queries $R_1,\dots,R_N$. 
Note that its rate is $\beta=\ell/M$ and locality of $\rx_i$ is $r_i=|R_i|/M$. 
We now choose positive integers $M'$ and $s$ such that
\begin{equation} \label{eq:app:lem:message_alphabet}
\frac{\log_2 |\Ac'|}{M \,\log_2|\Ac|} \leq \frac{s}{M'} \leq (1+\epsilon) \frac{\log_2 |\Ac'|}{M \,\log_2|\Ac|}
\end{equation}
 
We use $\Cs$ to design a valid index code $\Cs'$ for message alphabet $\Ac'$ with message length $M'$. For $\Cs'$, let $\p{y}_1,\dots,\p{y}_N \in {\Ac'}^{\,M'}$ denote the $N$ messages. 
From~\eqref{eq:app:lem:message_alphabet}, an injective function exists from ${\Ac'}^{\,M'}$ to $\Ac^{sM}$. 
For each $i \in [N]$, we use such a function to map $\p{y}_i$ to a vector $(\p{x}_i^{(1)},\dots,\p{x}_i^{(s)}) \in \Ac^{sM}$, where $\p{x}_i^{(j)} \in \Ac^M$ for all $j \in [s]$. 
With this, we represent a message $\p{y}_i$ over $\Ac'$ using $s$ generations of $M$-length vectors $\p{x}_i^{(j)}$ over $\Ac$.
The side information at $\rx_i$ is $\p{y}_n$, $n \in K_i$, which is equivalent to $\p{x}_n^{(j)}$, $n \in K_i$ and $j \in [s]$.

We transmit the $s$ generations of messages independently using the index code $\Cs$. The resulting index code $\Cs'$ has codelength $\ell' = s \, \ell$, and message length $M'$ over the message alphabet $\Ac'$. The rate of this code is 
\begin{align*}
\beta' = \frac{\ell' \log_2 |\Zc|}{M' \log_2 |\Ac'|} = \frac{s \ell \log_2 |\Zc|}{M' \log_2 |\Ac'|} &\leq (1 + \epsilon) \frac{\ell \log_2 |\Zc|}{M \log_2 |\Ac|} \\ 
&= (1 + \epsilon) \beta,
\end{align*} 
where we have used~\eqref{eq:app:lem:message_alphabet}.
The number of codeword symbols observed by $\rx_i$ for the code $\Cs'$ is $s |R_i|$ since $s$ instances of $\Cs$ are used to encode $\p{y}_1,\dots,\p{y}_N$. Thus the locality $r_i'$ of $\rx_i$ under the code $\Cs'$ is
\begin{equation*}
r_i' = \frac{s|R_i| \log_2 |\Zc|}{M' \log_2 |\Ac'|} \leq (1 + \epsilon)\frac{|R_i|\log_2 |\Zc|}{M \log_2 |\Ac|} = (1 + \epsilon)r_i.
\end{equation*}
We conclude that $\Cs'$ has rate at the most $(1+\epsilon)\beta$ and locality at the most $(1+\epsilon)r$. 
\end{proof}

If we choose $\epsilon$ and $\beta$ arbitrarily close to $0$ and $\beta_G^*(r,\Ac,\Zc)$, respectively, we observe from Lemma~\ref{lem:app:message_alphabet} that $\beta_G^*(r,\Ac',\Zc) \leq \beta_G^*(r,\Ac,\Zc)$.

\subsection{Summary} \label{app:sec:summary_invariance}

From the results of this appendix we observe that the value of $\beta_G^*(r,\Ac,\Zc)$ is independent of $\Ac$ and $\Zc$. Hence, for any choice of alphabet $\Ac$
\begin{equation*}
\beta_G^*(r,\Ac,\Ac) = \inf_{\Ac', \, \Zc'} \beta_G^*(r,\Ac',\Zc'),
\end{equation*} 
i.e., $\beta_G^*(r,\Ac,\Ac)$ captures the ultimate locality-rate trade-off even when optimization over message and channel alphabets is allowed.

\section{Proof of Lemma~\ref{lem:convexity}} \label{app:lem:convexity}

Assume $r_1,r_2 \geq 1$ and let $\epsilon > 0$. For each $j=1,2$, there exists an index code with broadcast rate $\beta_j \leq \beta_G^*(r_j) + \epsilon$ and locality at the most $r_j$. 
We will denote the blocklength of this code by $\ell_j$, message length by $M_j$ and the subsets of the indices used by the $N$ receivers as $R_{1,j},\dots,R_{N,j}$, where $j=1,2$. 
For some choice of non-negative integers $k_1$ and $k_2$, consider a time-sharing scheme where the first index code is used $k_1M_2$ times, and the second index code is used $k_2M_1$ times. 
For this composite scheme, the overall message length is $M=k_1M_2M_1 + k_2M_1M_2 = M_1M_2(k_1+k_2)$. The blocklength is $\ell=k_1M_2\ell_1 + k_2M_1\ell_2$, and the number of codeword symbols utilized by the $i^{\text{th}}$ receiver to decode its desired message is $|R_i|=k_1M_2|R_{i,1}| + k_2M_1|R_{i,2}|$. The locality of this time-sharing scheme can be upper bounded as
\begin{align*}
\max_i \frac{|R_i|}{M} &= \max_i \,\, \frac{k_1|R_{i,1}|}{(k_1+k_2)M_1} + \frac{k_2|R_{i,2}|}{(k_1+k_2)M_2} \\ 
&\leq \frac{k_1}{k_1+k_2} r_1 + \frac{k_2}{k_1+k_2}r_2.
\end{align*} 
Similarly, the broadcast rate $\beta$ of this time-sharing scheme can be shown to be equal to $k_1\beta_1/(k_1+k_2) + k_2\beta_2/(k_1+k_2)$, which is upper bounded as
\begin{align*}
\beta &= \frac{k_1}{k_1+k_2}\beta_1 + \frac{k_2}{k_1+k_2}\beta_2 \\
&\leq \frac{k_1}{k_1+k_2}\beta_G^*(r_1) + \frac{k_2}{k_1+k_2}\beta_G^*(r_2) + \epsilon.
\end{align*} 
Denoting $k_1r_1/(k_1+k_2) + k_2r_2/(k_1+k_2)$ by $r$, and by letting $\epsilon \to 0$, we observe that
\begin{equation*}
\beta_G^*(r) \leq \frac{k_1}{k_1+k_2}\beta_G^*(r_1) + \frac{k_2}{k_1+k_2}\beta_G^*(r_2).
\end{equation*}
Convexity follows by approximating any real number in the interval $(0,1)$ by a rational number $k_1/(k_1+k_2)$ to any desired accuracy by using sufficiently large $k_1$ and $k_2$. 

\section{Proof of Theorem~\ref{thm:zeroes}} \label{app:proof:thm:zeroes}

We will design a new encoding matrix $\pmb{L}'$ by modifying the subset of the columns of the given matrix $\pmb{L}$ corresponding to the column indices $\Sc_1 \cup \cdots \cup \Sc_N$. 
For the remaining indices $k \in \Mc_1 \cup \cdots \cup \Mc_N$, the $k^\tth$ columns of $\pmb{L}$ and $\pmb{L}'$ are equal, i.e., $\pmb{L}_k = \pmb{L}'_k$. For an arbitrary $i \in [N]$, we will now explain the construction of the column vectors $\pmb{L}'_k$, $k \in \Sc_i$. Since the symbols $\pmb{x}^\tr\pmb{L}'_k$, $k \in \Sc_i$, are queried only by $\rx_i$ and are unused by other receivers, we only need to consider the constraints that are imposed by the demands of $\rx_i$ while designing the column vectors $\pmb{L}'_k$, $k \in \Sc_i$.

We will introduce the notation which will be used in the rest of the proof. For any $E \subseteq [MN]$, let $U_{E} = \spp(\pmb{e}_k, k \in E)$, i.e., $U_E$ is the subspace of all vectors whose support is a subset of $E$. 
For any $F \subseteq [\ell]$, let $V_F = \spp(\pmb{L}_k, k \in F)$ and $V'_F = \spp(\pmb{L}'_k, k \in F)$.
Since $R_i = \Sc_i \cup \Mc_i$, we have $V_{R_i} = V_{\Sc_i} + V_{\Mc_i}$ and $V'_{R_i} = V'_{\Sc_i} + V'_{\Mc_i}$, where the addition corresponds to sum of subspaces.
From Theorem~\ref{thm:design_criterion} and using the fact that $\pmb{L}$ is a valid encoder matrix, we have $\pmb{e}_j \in V_{R_i} + U_{\Kb_i}$, for all $j \in \Db_i$, i.e., we have 
\begin{equation} \label{eq:lem:zeroes}
U_{\Db_i} = \spp(\pmb{e}_j, j \in \Db_i) \subseteq V_{R_i} + U_{\Kb_i} = V_{\Sc_i} + V_{\Mc_i} + U_{\Kb_i}.
\end{equation} 
Again using Theorem~\ref{thm:design_criterion}, we observe that $\pmb{L}'$ allows $\rx_i$ to decode its demand if and only if 
\begin{equation} \label{eq:lem:zeroes:2}
U_{\Db_i} \subseteq V'_{R_i} + U_{\Kb_i} = V'_{\Sc_i} + V'_{\Mc_i} + U_{\Kb_i}.
\end{equation} 

Using the validity of the encoder matrix $\pmb{L}$, we will first lower bound $|\Sc_i|$ which is the number of coded symbols queried uniquely by $\rx_i$.
From~\eqref{eq:lem:zeroes}, we obtain
\begin{align} \label{eq:lem:zeroes:3}
U_{\Db_i} = (V_{\Sc_i} + V_{\Mc_i} + U_{\Kb_i}) \cap U_{\Db_i}.
\end{align} 
Let $\Vin = (V_{\Mc_i} + U_{\Kb_i}) \cap U_{\Db_i}$ denote the subspace of $V_{\Mc_i} + U_{\Kb_i}$ contained in $U_{\Db_i}$, and let $\Vout$ be any subspace such that $\Vout \cap U_{\Db_i} = \{\pmb{0}\}$ and $\Vin + \Vout = V_{\Mc_i} + U_{\Kb_i}$. Continuing from~\eqref{eq:lem:zeroes:3}, we claim that
\begin{align}
U_{\Db_i} &= (V_{\Sc_i} + \Vin + \Vout) \cap U_{\Db_i} \label{eq:lem:zeroes:4} \\
&= \left( (V_{\Sc_i} + \Vout) \cap U_{\Db_i} \right) + \Vin. \label{eq:lem:zeroes:5}
\end{align} 
It is clear that the subspace in~\eqref{eq:lem:zeroes:4} contains the subspace in~\eqref{eq:lem:zeroes:5} since $\Vin \subseteq U_{\Db_i}$. To prove that~\eqref{eq:lem:zeroes:4} is contained in~\eqref{eq:lem:zeroes:5}, assume that $\pmb{y} \in V_{\Sc_i} + \Vout$ and $\pmb{z} \in \Vin$ are such that $\pmb{y} + \pmb{z} \in U_{\Db_i}$. Since $\pmb{z} \in \Vin \subseteq U_{\Db_i}$ and $\pmb{y} + \pmb{z} \in U_{\Db_i}$, we conclude that $\pmb{y} \in U_{\Db_i}$ as well. Thus $\pmb{y} \in (V_{\Sc_i} + \Vout) \cap U_{\Db_i}$, and hence, $\pmb{y} + \pmb{z} \in \left((V_{\Sc_i} + \Vout) \cap U_{\Db_i}\right) + \Vin$.

Considering the dimensions of the subspaces in~\eqref{eq:lem:zeroes:5}, we have 
\begin{equation} \label{eq:lem:zeroes:6}
\dim(U_{\Db_i}) \leq \dim((V_{\Sc_i} + \Vout) \cap U_{\Db_i}) + \dim(\Vin).
\end{equation} 
In order to proceed with the proof of the theorem, we will now show that $\dim((V_{\Sc_i} + \Vout)\cap U_{\Db_i}) \leq |\Sc_i|$. To do so, assume that $\pmb{y}_j + \pmb{z}_j$, $j=1,\dots,n$, form a basis for $(V_{\Sc_i} + \Vout)\cap U_{\Db_i}$ where $\pmb{y}_j \in V_{\Sc_i}$ and $\pmb{z}_j \in \Vout$ and $n=\dim((V_{\Sc_i} + \Vout)\cap U_{\Db_i})$. 
If the scalars $\alpha_1,\dots,\alpha_n$ are such that $\sum_j \alpha_j \pmb{y}_j = \pmb{0}$, then
\begin{equation*}
 \sum_{j = 1}^{n} \alpha_j(\pmb{y}_j + \pmb{z}_j) = \sum_{j=1}^{n} \alpha_j \pmb{z}_j \in \Vout.
\end{equation*} 
Since $\pmb{y}_j + \pmb{z}_j \in U_{\Db_i}$, we also observe that $\sum_j \alpha_j(\pmb{y}_j + \pmb{z}_j) \in U_{\Db_i}$. Using the fact $\Vout \cap U_{\Db_i} = \{\pmb{0}\}$, we deduce that $\sum_j \alpha_j(\pmb{y}_j + \pmb{z}_j)=\pmb{0}$, and hence, $\alpha_1=\cdots=\alpha_n=0$. We conclude that $\pmb{y}_1,\dots,\pmb{y}_n \in V_{\Sc_i}$ are linearly independent, and therefore
\begin{equation} \label{eq:lem:zeroes:7}
\dim((V_{\Sc_i} + \Vout) \cap U_{\Db_i}) = n \leq \dim(V_{\Sc_i}) \leq |\Sc_i|.
\end{equation} 

Note that the columns of $\pmb{L}$ and $\pmb{L}'$ corresponding to the column indices $\Mc_i$ are equal, and hence $V_{\Mc_i}'=V_{\Mc_i}$. Using this fact together with~\eqref{eq:lem:zeroes:6} and~\eqref{eq:lem:zeroes:7}, we have
\begin{align*}
|\Sc_i| &\geq \dim((V_{\Sc_i} + \Vout)\cap U_{\Db_i}) \\
        &\geq \dim(U_{\Db_i}) - \dim(\Vin) \\
        &= \dim(U_{\Db_i}) - \dim((V_{\Mc_i} + U_{\Kb_i}) \cap U_{\Db_i}) \\
        &= \dim(U_{\Db_i}) - \dim((V'_{\Mc_i} + U_{\Kb_i}) \cap U_{\Db_i}).
\end{align*} 

From~\eqref{eq:lem:zeroes:2}, in order to satisfy the claim of this theorem, it is sufficient to chose the $|\Sc_i|$ vectors $\pmb{L}'_k$, $k \in \Sc_i$, such that $\pmb{L}'_k \lhd \Db_i$, i.e., $\pmb{L}'_k \in U_{\Db_i}$ and 
\begin{equation*}
 V'_{\Sc_i} \,+\, \left( (V'_{\Mc_i} + U_{\Kb_i}) \cap U_{\Db_i} \right) \supseteq U_{\Db_i}.
\end{equation*} 
This is always possible since the difference in the dimensions of $U_{\Db_i}$ and $(V'_{\Mc_i} + U_{\Kb_i}) \cap U_{\Db_i}$ is at the most $|\Sc_i|$. 
One way to construct $\pmb{L}'_k$, $k \in \Sc_i$, is as follows. We begin with a basis for $(V'_{\Mc_i} + U_{\Kb_i}) \cap U_{\Db_i}$. These vectors form a linearly independent set in $U_{\Db_i}$. We extend this set to a basis for $U_{\Db_i}$. The number of additional vectors in this basis is $\dim(U_{\Db_i}) - \dim((V'_{\Mc_i} + U_{\Kb_i}) \cap U_{\Db_i})$. These additional vectors together with $|\Sc_i|-\dim(U_{\Db_i}) + \dim((V'_{\Mc_i} + U_{\Kb_i}) \cap U_{\Db_i})$ all-zero vectors are chosen as the columns $\pmb{L}'_k$, $k \in \Sc_i$.

\section{Construction of the Bipartite Graph $\Bc$} \label{app:construct_bipartite}

We require the following lemma to prove the correctness of the proposed construction.
Note that $r_1=|R_1|,\dots,r_N=|R_N|$ are the degrees of the vertices $u_1,\dots,u_N$, respectively, in $\Bc$, and are the localities of the $N$ receivers in the index coding problem.
\begin{lemma} \label{lem:first_i_entries_of_r}
For any $i \in [N]$, $r_1 + \cdots + r_i \leq 2i-1$.
\end{lemma}
\begin{proof}
Let $r_1 + \cdots + r_i = y$. The average of the first $i$ entries of $\p{r}$ is $y/i$. Since the components of $\p{r}$ are in the ascending order, we have $r_{i+1},r_{i+2},\dots,r_N \geq y/i$. Using this in 
\begin{align*}
y + r_{i+1} + r_{i+2} + \cdots + r_N = r_1 + \cdots + r_N = 2(N-1), 
\end{align*} 
we have $y + (N-i)y/i \leq 2(N-1)$. This implies $y \leq 2(N-1)i/N < 2i$, and since $y$ is an integer, we deduce that $y \leq 2i-1$.
\end{proof}

\subsection*{Proof of Lemma~\ref{lem:construction_Bc}}

The construction of $\Bc$ starts with the vertex sets $\{u_1,\dots,u_N\}$ and $\{v_1,\dots,v_{N-1}\}$, and an empty set of edges.
Edges are appended sequentially by considering the vertices $u_1,\dots,u_N$, in that order. At step $i$, we determine all the neighbors of $u_i$, while making sure that the degrees of $v_1,\dots,v_{N-1}$ do not exceed $2$. 
Using $i=1$ in Lemma~\ref{lem:first_i_entries_of_r}, we deduce that $r_1$ is necessarily equal to $1$. We append $\{u_1,v_1\}$ to the edge set, and this completes the first step. 
For $2 \leq i \leq N-1$, at step $i$, we append $\{u_i,v_i\}$ to the edge set of $\Bc$, and identify $r_i-1$ vertices among $\{v_1,\dots,v_{i-1}\}$ with degrees equal to $1$, and include edges between $u_i$ and these $r_i-1$ vertices. We will shortly show that such a set of $r_i-1$ vertices indeed exist.
This procedure ensures that degrees of $v_1,\dots,v_{N-1}$ do not exceed $2$, and for any $i \leq N-1$, we have $i \in R_i$ and $R_i \subseteq [i]$. 
Finally, at step $N$, we append edges between $u_N$ and all the vertices among $\{v_1,\dots,v_{N-1}\}$ that have degree $1$.
Since $\sum_{i \in [N]}r_i=2(N-1)$, at the end of step $N$, it is guaranteed that the degrees of $u_1,\dots,u_N$ are $r_1,\dots,r_N$, respectively, and that the degrees of $v_1,\dots,v_{N-1}$ are all equal to $2$.

It only remains to show that at step $i$, for $2 \leq i \leq N-1$, there exist $r_i-1$ vertices among $v_1,\dots,v_{i-1}$ with degree equal to $1$. We prove this by induction. 
Let us first consider $i=2$. We know from Lemma~\ref{lem:first_i_entries_of_r} that $r_1=1$ and $r_1 + r_2 \leq 3$, i.e, $r_2 \leq 2$ and $r_2-1 \leq 1$.
At the end of step $1$, $v_1$ has degree $1$, which is because of the edge $\{u_1,v_1\}$. Since at the most $1$ vertex with degree $1$ is required, step 2 can be completed.
Now consider $3 \leq i \leq N-1$, and assume that steps $1,\dots,(i-1)$ have been completed. 
Now consider the end of step $(i-1)$. At this point, the edges $\{u_1,v_1\},\dots,\{u_{i-1},v_{i-1}\}$ are present in $\Bc$. Thus the degrees of each of $v_1,\dots,v_{i-1}$ is either $1$ or $2$. Let the number of vertices with degree $1$ among $v_1,\dots,v_{i-1}$ at the end of step $(i-1)$ be $z$. 
Thus the sum of the degrees of $v_1,\dots,v_{i-1}$ is $z + 2(i-1-z)=2i-z-2$.
Since the neighbors of each of $v_1,\dots,v_{i-1}$ are subsets of $u_1,\dots,u_{i-1}$ at the end of step $(i-1)$, and since the degrees of $u_1,\dots,u_{i-1}$ are $r_1,\dots,r_{i-1}$ at this stage, we have $2i-z-2 = r_1 + \cdots + r_{i-1}$. Using Lemma~\ref{lem:first_i_entries_of_r}, we have $2i-1 \geq r_1 + \cdots + r_i = 2i-z-2 + r_i$ which implies $z \geq r_i-1$. This proves that step $i$ can be completed.

\section{Proof of Lemma~\ref{lem:cyclic_symmetry}} \label{app:lem:cyclic_symmetry}

Since $\sigma$ is an automorphism of $G$, so is $\sigma^n$ for any $n \in [N]$. Note that the group $\{\sigma,\sigma^2,\dots,\sigma^{N}=1\}$ acts transitively on the vertex set of $G$.
Let $(\Ef,\Df_1,\dots,\Df_N)$ be an index code for $G$ with rate $\beta$, message length $M$, receiver localities $r_1,\dots,r_N$, and receiver queries $R_1,\dots,R_N$. 
Note that $r \geq \ravg = \sum_{i \in [N]} r_i/N$.
We will consider $N$ index coding schemes $(\Ef^{(n)},\Df_1^{(n)},\dots,\Df_N^{(n)})$, $n \in [N]$, each of which is derived from $(\Ef,\Df_1,\dots,\Df_N)$ by permuting the roles of the messages $\xb_1,\dots,\xb_N$. Specifically, the $n^{\text{th}}$ encoder $\Ef^{(n)}$ is the encoder $\Ef$ applied to the $n^{\text{th}}$ left cyclic shift of the message tuple $\xb_1,\dots,\xb_N$, i.e.,
\begin{align*}
\Ef^{(n)}\left( \, (\xb_1,\dots,\xb_N) \,\right) \!&=\! \Ef( (\xb_{n+1},\xb_{n+2},\dots,\xb_N,\xb_1,\dots,\xb_{n}) ) \\ 
&= \Ef\left( \, (\xb_{\sigma^{n}(1)},\dots,\xb_{\sigma^{n}(N)}) \, \right).
\end{align*}
For any $i \in [N]$, in the above expression of $\Ef^{(n)}$, the message $\xb_i$ is the $(i-n)_N^{\text{th}}$ argument of $\Ef$ where $(i-n)_N=(i-n)$ if $(i-n) \geq 1$ and $(i-n)_N = i-n+N$ otherwise. Hence, the encoding function $\Ef^{(n)}$ operates on the message $\xb_i$ in the same manner as $\Ef$ operates on the message $\xb_{(i-n)_N}$. Using the fact that $\sigma$ is an automorphism of $G$, it is easy to see that, when the $n^{\text{th}}$ code is used, the $i^{\text{th}}$ receiver can decode $\xb_i$ as $\Df_{(i-n)_N}(\cb_{R_{(i-n)_N}},\xb_{K_i})$. Thus, the number of codeword symbols queried by the $i^{\text{th}}$ receiver in the $n^{\text{th}}$ code is $|R_{(i-n)_N}| = M\,r_{(i-n)_N}$, where $M$ is the message length.

Now consider a time sharing scheme that utilizes each of the $N$ encoders $\Ef^{(1)},\dots,\Ef^{(N)}$ exactly once. The overall message length for this scheme is $MN$, the broadcast rate is $\beta$, and the number of codeword symbols queried by the $i^{\text{th}}$ receiver is 
\begin{equation*}
|R_i'| = \sum_{n \in [N] }|R_{(i-n)_N}| = M \sum_{ n \in [N]} r_n = MN \ravg,
\end{equation*}
which is independent of $i$.
Also, the overall locality of this time sharing scheme is $r' = \max_i |R_i'|/MN = \ravg \leq r$.
To complete the proof observe that for any $i \in [N]$ and $j = i \mod N + 1$, we have
\begin{align*}
|R_i' \cap R_j'| &= \sum_{n \in [N]} |R_{(i-n)_N} \cap R_{(i-n+1)_N}| \\
                 &= \sum_{n \in [N]} |R_{n} \cap R_{(n+1)_N}|,
\end{align*}  
which is independent of $i$.


\section*{Acknowledgment}

The first author thanks Prof.\ B.\ Sundar\ Rajan for discussions regarding the topic of this paper, as well as Dr.\ Yi Hong and Prof.\ Emanuele Viterbo for hosting him at Monash University which led to several insightful discussions.



\end{document}